\documentclass[sigconf, nonacm]{acmart}

\newcommand\vldbdoi{XX.XX/XXX.XX}
\newcommand\vldbpages{XXX-XXX}
\newcommand\vldbvolume{14}
\newcommand\vldbissue{1}
\newcommand\vldbyear{2020}
\newcommand\vldbauthors{\authors}
\newcommand\vldbtitle{\shorttitle} 
\newcommand\vldbavailabilityurl{}
\newcommand\vldbpagestyle{plain}

\usepackage{microtype}
\usepackage{graphicx}
\usepackage{subfigure}
\usepackage{booktabs} 

\usepackage{datetime}
\usepackage{url}
\usepackage[ruled, noend, vlined]{algorithm2e}

\usepackage{amsmath}
\usepackage{algpseudocode}
\algrenewcommand\algorithmicrequire{\textbf{Input:}}
\algrenewcommand\algorithmicensure{\textbf{Output:}}
\newcommand{\dist}{\operatorname{dist}}

\usepackage{float}
\usepackage{color}
\usepackage{multirow}
\usepackage{soul}
\usepackage{array}
\usepackage{subcaption}
\usepackage{soul}
\usepackage{balance}
\usepackage{enumitem}
\usepackage{graphicx}
\usepackage{kotex}
\usepackage{mathtools}
\usepackage{physics}

\usepackage{stmaryrd}
\usepackage{trimclip}
\usepackage{xcolor}
\usepackage{fontawesome}
\usepackage{needspace}

\usepackage{amsfonts}

\usepackage{pifont}

\makeatletter
\DeclareRobustCommand{\shortto}{%
  \mathrel{\mathpalette\short@to\relax}%
}

\newcommand{\short@to}[2]{%
  \mkern2mu
  \clipbox{{.3\width} 0 0 0}{$\m@th#1\vphantom{+}{\shortrightarrow}$}%
  }
\makeatother

\newcolumntype{L}[1]{>{\raggedright\let\newline\\\arraybackslash\hspace{0pt}}m{#1}}

\usepackage{tikz}
\usetikzlibrary{calc}

\newtheorem{theorem}{\bf Theorem}

\newtheorem{definition2}{\bf Definition}
\newtheorem{lemma2}{\bf Proposition}

\newtheorem{p-rule}{\bf Rule}

\usepackage{hyperref}

\let\oldnl\nl
\newcommand{\nonl}{\renewcommand{\nl}{\let\nl\oldnl}}

\newif\ifshowrevisions
\showrevisionsfalse

\definecolor{colorA}{HTML}{88CCEE}
\definecolor{colorB}{HTML}{CC6677}
\definecolor{colorC}{HTML}{DDCC77}
\definecolor{colorD}{HTML}{117733}
\definecolor{colorE}{HTML}{332288}
\definecolor{colorF}{HTML}{AA4499}
\definecolor{colorG}{HTML}{44AA99}
\definecolor{colorH}{HTML}{999933}
\definecolor{colorI}{HTML}{882255}
\definecolor{colorJ}{HTML}{661100}
\definecolor{colorK}{HTML}{6699CC}

\newcommand{\reviewColorA}{colorB}

\newcommand{\reviewnote}[2]{\colorbox{white}{%
\begin{minipage}{32pt}%
{\color{#1}#2}%
\end{minipage}%
}}
\usepackage{marginnote}

\ifshowrevisions
  \newcommand{\revisionstyle}[2][\reviewColorShared]{{\color{#1}#2}}
  \newcommand{\request}[2]{%
    \marginnote{\reviewnote{#1}{\textsf{\textbf{#2}}}}%
  }
\else
  \newcommand{\revisionstyle}[2][\reviewColorShared]{#2}   
  \newcommand{\request}[2]{}                               
\fi

\newcommand{\revisionBase}[3]{\request{#1}{#2}{\revisionstyle[#1]{#3}}}

\newcommand{\revisionA}[2]{\revisionBase{\reviewColorA}{#1}{#2}}

\newcommand{\FuncName}[1]{\textsc{{#1}}}

\usepackage{booktabs}
\usepackage{tabularx}
\usepackage{pifont}
\newcommand{\cmark}{\ding{51}}
\newcommand{\xmark}{\ding{55}}

\newif\ifFullVersion

\def\FullVersion{\let\ifFullVersion=\iftrue}
\def\ShortVersion{\let\ifFullVersion=\iffalse}
\FullVersion

\newcommand{\XJoin}{\FuncName{{XJoin}}\xspace}
\newcommand{\SimJoin}{\FuncName{{SimJoin}}\xspace}
\newcommand{\DiskJoin}{\FuncName{{DiskJoin}}\xspace}

\newcommand{\RoarGraph}{\FuncName{{RoarGraph}}\xspace}

\newcommand{\VectorJoin}{\FuncName{VectorJoin}}

\newcommand{\SIFTE}{\FuncName{SIFT-E}\xspace}
\newcommand{\SIFT}{\FuncName{SIFT}\xspace}
\newcommand{\GIST}{\FuncName{GIST}\xspace}
\newcommand{\Glove}{\FuncName{GloVe}\xspace}
\newcommand{\NYTimes}{\FuncName{NYTimes}\xspace}
\newcommand{\FMNIST}{\FuncName{FMNIST}\xspace}
\newcommand{\COCO}{\FuncName{COCO}\xspace}
\newcommand{\ImageNet}{\FuncName{ImageNet}\xspace}
\newcommand{\LAION}{\FuncName{LAION}\xspace}

\newcommand{\ES}{\FuncName{INLJ}\xspace}

\newcommand{\ESMIA}{\FuncName{MI+Adapt}\xspace}

\newcommand{\NLJ}{\FuncName{NLJ}\xspace}
\newcommand{\INLJ}{\FuncName{INLJ}\xspace}

\newcommand{\HWS}{\FuncName{HWS}\xspace}
\newcommand{\SWS}{\FuncName{SWS}\xspace}
\newcommand{\MI}{\FuncName{MI}\xspace}

\newcommand{\MIA}{\FuncName{MI+Adapt}\xspace}
\newcommand{\MISWS}{\FuncName{MI+SWS}\xspace}

\newcommand{\SWSH}{\FuncName{SWS+Hybrid}\xspace}
\newcommand{\MIH}{\FuncName{MI+Hybrid}\xspace}

\SetKw{Continue}{continue}
\SetKw{Break}{break}

\newcommand{\CASE}[1]{\STATE \textbf{case} #1\textbf{:} \begin{ALC@g}}
\newcommand{\ENDCASE}{\end{ALC@g}}

\newcommand{\DEFAULT}{\STATE \textbf{default:} \begin{ALC@g}}
\newcommand{\ENDDEFAULT}{\end{ALC@g}}
\newcommand{\DEFAULTLINE}[1]{\STATE \textbf{default:} }

\newcommand{\bluecomment}[1]{}
\newcommand{\redcomment}[1]{\tcc{#1}}

\newcommand{\nata}[1]{#1}
\newcommand{\sigmod}[1]{#1}

\usepackage{etoolbox}

\definecolor{colorMoved}{HTML}{00695C}   
\definecolor{colorPlanned}{HTML}{6A4FA3} 
\DeclareRobustCommand{\moved}[1]{{\color{colorMoved}#1}}
\DeclareRobustCommand{\planned}[1]{{\color{colorPlanned}#1}}

\newcommand{\todo}[1]{#1}

\newcommand{\kkim}[1]{#1}

\ifFullVersion

\else

\fi

\LinesNumbered
\SetAlgoCaptionSeparator{.}
\SetKwProg{Fn}{Function}{}{end}
\SetKwFor{uForEach}{foreach}{do}{}
\SetStartEndCondition{ (}{) }{)}
\SetNlSty{texttt}{}{:}
\SetArgSty{}

\begin{document}


\title{Fast Approximate Vector Joins via Offline-Online Co-Design}

\author{Kyoungmin Kim}
\affiliation{%
  \institution{EPFL}
}
\email{kyoung-min.kim@epfl.ch}

\author{Lennart Roth}
\affiliation{%
  \institution{EPFL}
}
\email{lennart.roth@epfl.ch}

\author{Liang Liang}
\affiliation{%
  \institution{EPFL}
}
\email{liang.liang@epfl.ch}

\author{Anastasia Ailamaki}
\affiliation{%
  \institution{EPFL}
}
\email{anastasia.ailamaki@epfl.ch}


\begin{abstract}
\revisionA{}{
Approximate vector join is a fundamental primitive for AI-driven data management. Prior methods reduce online search cost through work sharing across queries, but remain limited by hard reuse of only final join results, separate query/data-side indexes over vectors that already live in a shared embedding space, and recall loss on difficult queries that lie far from data distributions. We present an offline–online co-design for approximate vector join that addresses these limitations through three complementary ideas. First, we introduce an encoder-centric merged index that jointly organizes vectors from the same embedding space and offloads the main bottleneck of vector join to offline index construction.
Second, we propose soft work sharing that generalizes reuse from final join results to useful search effort. 
Third, we develop an adaptive hybrid search that selectively improves recall for difficult queries while avoiding unnecessary overhead on easy ones. 
Together, these techniques provide \sigmod{a modular yet principled design} that improves both efficiency and robustness across diverse datasets. By treating the underlying graph index as a black box, it can generalize across arbitrary graph indexes.
}


\end{abstract}

\maketitle

\pagestyle{\vldbpagestyle}

\ifdefempty{\vldbavailabilityurl}{}{
\vspace{.3cm}
\begingroup\small\noindent\raggedright\textbf{PVLDB Artifact Availability:}\\
The source code, data, and/or other artifacts have been made available at \url{\vldbavailabilityurl}.
\endgroup
}

\vspace*{-0.2cm}
\section{Introduction}\label{sec:intro}

Similarity joins over vector embeddings are a fundamental operation in modern AI-driven data management and analytics \cite{wang2013scalable, pansurvey}.
Many real-world applications require identifying all pairs of semantically similar objects across two collections, rather than retrieving only the top-$k$ nearest neighbors for each query \cite{DiskJoin}. 
\revisionA{R2.D4}{An example is cross-source product matching, where two catalogs are encoded by the same embedding model and threshold-based vector join retrieves all candidate pairs (detailed in Appendix A) \cite{li2020deep, tracz2020bert}.}
Near-duplicate detection in image, video, or document collections relies on self-joins to find all items whose embeddings fall within a small distance threshold \cite{wang2013scalable}. 
In data integration and record linkage, vector joins are used to match entities across heterogeneous sources, such as joining product descriptions, customer profiles, or scientific records based on semantic similarity \cite{clusterjoin, wang2024xling}. 
In recommendation, fraud detection, and security analytics, vector joins enable discovering clusters of closely related behaviors or anomalously similar activity patterns \cite{wang2013scalable}. 
They also underlie hybrid vector-relational queries, enabling filters, joins, and aggregations over both symbolic and semantic attributes \cite{zhang2023vbase, sanca2024efficient, chen2024singlestore, pgvector, chronis2025filtered}.

Given two sets of vectors $X$ and $Y$ (assume $|X| \leq |Y|$) and a distance threshold $\theta$, the goal of threshold-based vector join is to find all similar vector pairs $(x,y)$ such that $x \in X, y \in Y, dist(x, y) < \theta$ for a distance function $dist(\cdot,\cdot)$. 
\revisionA{}{We focus on the \emph{in-memory} vector-join setting, where the primary bottleneck is distance computation over high-dimensional vectors rather than disk I/O \cite{ootomo2024cagra, DiskJoin}. Accordingly, our goal is to reduce the amount of online distance computations for candidate pairs.

A straightforward way to evaluate threshold-based vector joins is to probe an approximate nearest neighbor (ANN) index over the larger set $Y$ (simply data) independently for each vector in $X$ (simply query), i.e., an index nested-loop style execution \cite{DuckDBV, pgvector, chen2024singlestore}.
\sigmod{As in top-$k$ search, graph-based indexes such as HNSW \cite{HNSW} and NSG \cite{NSG} have been used due to their impressive performance \cite{wang2024xling, SimJoin}.} 
Then, for each query, a threshold-based search first reaches an \emph{in-range region}, a data region containing points within distance $\theta$ of the query (via greedy search), and then expands within that region to collect all matches (post-greedy expansion) \cite{SimJoin}.

\sigmod{We observe that the greedy search can take up to 99\% of the total time especially under small thresholds (Figure \ref{fig:profiling_exp}), traversing a large portion of the index before hitting an in-range point.
Moreover, each search starts from a designated index entry, resulting in \emph{redundant} traversals that could have been reused for similar queries.
Recent work \cite{SimJoin} aims to mitigate this} through \emph{work sharing}: similar queries are processed in a certain order \sigmod{(obtained from a query-side index)} so that a later query can reuse information from an earlier one. However, this has been only done in \emph{hard} form, reusing only the final {in-range} points found as the seeds to later query traversals. This makes reuse restrictive, since useful effort spent reaching nearby but out-of-range (simply out-range) points is discarded. 
\sigmod{Moreover, since query- and data-side indexes are maintained separately, the online phase must rediscover the proximity relationships between the two vector sets from scratch for every join operation, incurring repetitive overhead.}
Finally, even after reaching an in-range region, traversing only within that region \cite{SimJoin} may miss other disconnected in-range regions and lose recall. This is pronounced for out-of-distribution (OOD) queries where the query and data distributions barely overlap, and closest data points to a query are scattered over the data space \cite{chen2024roargraph}.



\sigmod{We can then summarize these into three orthogonal issues:}
(1) \emph{redundancy across similar queries} \sigmod{that has been addressed only in a limited way \cite{SimJoin},} (2) \emph{redundancy across joins} \sigmod{that repeatedly rediscovers query-data relationships,} 
and (3) \emph{reachability failure on difficult queries} that, even after reaching one in-range region, the search may fail to reach other disconnected in-range regions.

This paper addresses these issues through a unified offline–online co-design, \sigmod{organized around a single principle: shift as much join effort as possible to offline index construction, and handle only what cannot be anticipated offline at query time, as efficiently and robustly as possible, while remaining simple and index-agnostic.
This principle naturally gives rise to three components. 
To mitigate (2), we propose a merged index (offline) that organizes vectors from a shared embedding space in one unified graph, eliminating the need to rediscover cross-set geometry at query time and offloading initial in-range region discovery to index construction in a \emph{threshold-agnostic} way. 
To mitigate (1), we introduce soft work sharing (online), which generalizes reuse from only final join results to useful search effort produced during traversal (closest data point found so far, which may be out-range). Importantly, soft work sharing remains beneficial even when merged indexing is unavailable as it directly reduces redundant online traversal across similar queries regardless of index organization. 
To mitigate (3), we propose an OOD-aware adaptive hybrid search (online) that traverses out-range points as bridges to disconnected in-range regions, but only for queries predicted OOD using a lightweight query-data overlap signal provided by the merged index. 
Each component addresses what the previous one leaves unsolved: the merged index eliminates redundant rediscovery of cross-set geometry, soft work sharing maximizes reuse of remaining online traversal, and hybrid search recovers the recall that neither preparation nor reuse can guarantee.
We integrate them in a unified execution framework that captures existing state of the art \cite{SimJoin}, enables modular design, fair comparison, and clean ablations across a well-defined design space.}
In summary, our contributions are as follows.

\begin{itemize}[topsep=0pt, leftmargin=0.5cm]

\item \nata{We argue that the current bottleneck in vector joins is not just search efficiency, but a misalignment between the data's natural embedding geometry and the traditional table-centric index organization. We propose a unified offline–online co-design that resolves this misalignment; by capturing \SimJoin \cite{SimJoin} as a special case within the same pipeline, the framework also enables principled comparisons and clean ablations across the full design space\moved{, satisfying all design criteria in Table \ref{tab:novelty} for the first time}.}

\item \nata{We introduce soft work sharing, which redefines \emph{reusable work} as traversal effort rather than just result sets, reducing latency by up to 3x over hard work sharing while also cutting memory usage for large thresholds.}

\item \nata{We propose merged index, an offline encoder-centric strategy that treats the embedding space as the primary join substrate, effectively reducing initial join search to a constant-time neighborhood lookup, reducing latency by up to 44.1x.}

\item \nata{We identify disconnected in-range regions as a structural recall failure mode for OOD queries, and introduce an OOD-aware adaptive hybrid search that selectively bridges these regions — improving recall by up to 43\% on OOD datasets while imposing no overhead on in-distribution queries.}





\item We validate the proposed design through comprehensive experiments. Across eight datasets, two graph indexes, and scalability tests up to 10M vectors, the results show that \sigmod{our three techniques are simple} yet effective and complementary building blocks for efficient approximate vector join.
\end{itemize}

}

\revisionA{}{
In the remainder of the paper, Section \ref{sec:related} explains related work, Section \ref{sec:vector_joins} defines the vector join problem and explains basic approaches, Section \ref{sec:ours} describes our \sigmod{unified} vector join framework and techniques, Section \ref{sec:exp} presents experimental study, and Section \ref{sec:conclusion} concludes the paper.
}

\section{Related Work}\label{sec:related}

\revisionA{}{
Prior work on vector search and vector join covers different parts of the overall design space, but does not provide a unified solution to approximate threshold-based vector join with broad work reuse, encoder-centric indexing, and selective handling of difficult OOD queries. Table~\ref{tab:novelty} summarizes the main distinctions. Existing methods either focus on single-query ANN search (ANNS), reduce unnecessary query searches without cross-query reuse, reuse only hard in-range results under separate query/data indexes, optimize disk-based self-join and I/O reuse, or address OOD search through a specialized query-guided index rather than vector join. 
In contrast, our framework exposes a broader offline--online co-design space over existing graph indexes, rather than inventing yet another bespoke index structure that only works in narrow settings. 
To the best of our knowledge, we are the first to support fast threshold-based vector join with soft cross-query work sharing, encoder-centric merged indexing, and robust accuracy with OOD-aware adaptive traversal in one unified framework.}
\nata{The absence of this combination is not accidental. Query- and data-side indexes evolved from the relational table abstraction, which treats each vector set as a separate entity. This table-centric framing made it conceptually natural to build separate indexes, obscuring the fact that shared-encoder vectors already inhabit a common geometric substrate.}


\begin{table*}[t]
\centering
\small
\setlength{\tabcolsep}{4pt}
\caption{\revisionA{Novelty}{Comparison of prior methods and ours along the key design dimensions of approximate threshold-based vector join.}}
\vspace*{-0.2cm}
\label{tab:novelty}
\begin{tabular}{lccccccc}
\toprule
\textbf{Method} &
\textbf{Thr.-} &
\textbf{Opt for reducing} &
\textbf{Work} &
\textbf{OOD} &
\textbf{Query awareness in} &
\textbf{Index / data} &
\textbf{Generalizability} \\
&
\textbf{join} &
\textbf{dist. computations} &
\textbf{sharing} &
\textbf{handling} &
\textbf{index construction} &
\textbf{organization} &
\textbf{across indexes} \\
\midrule
Top-$k$ ANNS     & \xmark & \cmark & \xmark & \xmark & \xmark\ (ad-hoc queries) & Data-side index only & \xmark \\
\XJoin \cite{wang2024xling}            & \cmark & \cmark & \xmark & \xmark & \xmark & Data-side index only & \cmark \\
\SimJoin \cite{SimJoin}           & \cmark & \cmark & \cmark & \xmark & For query-side index only & Separate query/data index & \cmark \\
\DiskJoin \cite{DiskJoin}         & \cmark & \xmark & \cmark (I/O reuse) & \xmark & \xmark\ (no separate queries) & Self-join-only index & \xmark \\
\RoarGraph \cite{chen2024roargraph}        & \xmark & \cmark & \xmark & \cmark & From sampled queries only & Query-guided data index & \xmark \\
\midrule
\textbf{Ours}    & \textbf{\cmark} & \textbf{\cmark} & \textbf{\cmark} & \textbf{\cmark} & \textbf{Jointly with data} & \textbf{Merged from construction} & \textbf{\cmark} \\
\bottomrule
\vspace*{-0.5cm}
\end{tabular}
\vspace*{-0.2cm}
\end{table*}

\vspace*{-0.2cm}
\subsection{Non-Graph-Based ANN Index}

A wide range of index structures has been developed for ad-hoc queries. 
Early research was dominated by Locality-Sensitive Hashing (LSH) \cite{datar2004locality, gan2012locality}, which maps similar vectors into hash buckets with high probability, so queries only probe a few buckets instead of scanning all data \cite{jafari2021survey}.
However, this approach is not effective for high-dimensional vectors due to curse of dimensionality \cite{bohm2001searching}; errors in hashing similar vectors into the same bucket increase with the dimension.
Inverted file (IVF) and quantization-based approaches cluster the vector space into many buckets and search only a few relevant buckets per query \cite{baranchuk2018revisiting, jegou2010product, guo2020accelerating}.
Such ANN indexes optimized for single-query retrieval are incorporated into many vector processing systems and libraries (e.g., FAISS \cite{douze2024faiss} and Milvus \cite{wang2021milvus}), including the recent graph indexes below.

\vspace*{-0.2cm}
\subsection{Graph-Based ANN Index}


\revisionA{}{
Graph-based ANN indexes such as HNSW \cite{HNSW}, Vamana \cite{DiskANN}, and NSG \cite{NSG} are the dominant substrate for high-performance vector search due to their strong efficiency--accuracy trade-off \cite{HNSW, NSG, douze2024faiss}. They are, however, designed primarily for single-query top-$k$ retrieval. In-memory graph traversal is typically efficient enough that the main bottleneck becomes repeated distance computation over high-dimensional vectors, rather than disk I/O, which is exactly the regime we target in this paper. 
\revisionA{R1.D3}{Existing graph-based ANN work also includes techniques for adaptive termination and query-specific search control. For example, DARTH \cite{DARTH} proposes ML-based early stopping by estimating recall during search, but it assumes a fixed-start single-query ANN regime; this is not directly portable to our work-sharing setting, where starting points vary dynamically across queries. \sigmod{It also targets top-$k$ queries only, not threshold-based ones.
}}

Recent work has also begun to address OOD queries in graph search \cite{chen2024roargraph, jaiswal2022ood, VIBE}. \RoarGraph \cite{chen2024roargraph} incorporates sampled historical queries into index construction through a specialized query-guided graph, improving OOD retrieval for ANN search. Our setting is different in two ways: we study threshold-based vector join rather than single-query top-$k$ search, and 
we use online hybrid search over existing graph indexes rather than building a dedicated index.
}


\vspace*{-0.2cm}
\subsection{Approximate Vector Join}

\revisionA{}{
The most basic execution model for vector join is analogous to nested-loop or index nested-loop join, exploited by many vector systems including pgvector \cite{pgvector}, VBASE \cite{zhang2023vbase}, SingleStore-V \cite{chen2024singlestore}, DuckDB \cite{DuckDBV}, and more \cite{wang2021milvus, pinecone, qdrant}. A naive nested-loop join computes all pairwise distances exactly, which is prohibitive at scale. Index nested-loop style execution probes a data-side ANN index independently for each query and is therefore much faster, but it repeats traversal even when many queries are similar.

\XJoin \cite{wang2024xling} is an early attempt to reduce this overhead. Rather than sharing work across queries, it filters queries that are unlikely to have enough results. Thus, it still performs its ANN traversal largely independently and may degrade recall.
\SimJoin \cite{SimJoin} advances this line of work. It introduces cross-query work sharing by ordering queries through a minimum spanning tree (MST) over a query-side graph index and reusing the join results of an executed query as starting seeds for its descendants. It significantly outperforms \XJoin by reducing redundant traversal across similar queries rather than only pruning which queries to execute. However, its reuse is still restrictive: it caches only in-range points and discards useful traversal effort spent reaching nearby but out-range points. 
\sigmod{Moreover, because it maintains separate query- and data-side indexes, the online phase must rediscover the proximity relationships between the two sets for every join operation.} 

\DiskJoin \cite{DiskJoin} studies a different setting of SSD-resident self-join. Its main contribution is I/O reuse through bucket scheduling and disk-aware processing, rather than in-memory graph traversal and cross-query reuse. We therefore view \DiskJoin as complementary to our work rather than a direct baseline for in-memory setup.
}




There are also distributed and parallel solutions for similarity joins, such as MAPSS \cite{wang2013scalable} and ClusterJoin \cite{clusterjoin} which are MapReduce-based join frameworks. 
\todo{It is an interesting future work to extend our framework to disk-based, distributed, and parallel environments.}

\vspace*{-0.2cm}
\subsection{Softening Hard Decisions in ML}

\revisionA{Novelty}{Our soft work sharing was inspired by prior work in ML, which suggests that replacing a hard mechanism with a soft one can yield substantial practical gains, while still being conceptually simple. Label smoothing \cite{DBLP:conf/cvpr/SzegedyVISW16, DBLP:conf/cvpr/ZhongC0J21} softens one-hot labels and improves generalization and calibration, knowledge distillation \cite{DBLP:journals/corr/HintonVD15, DBLP:conf/icml/ChandrasegaranT22} replaces hard supervision with soft teacher targets, allowing the student to exploit richer inter-class similarity information than hard labels alone, and Gumbel-Softmax \cite{DBLP:conf/iclr/JangGP17} replaces hard categorical choices with a differentiable relaxation that substantially improves optimization for discrete decisions \cite{DBLP:conf/iclr/MaddisonMT17}. 
\nata{These examples support our intuition: just as soft labels let a learner benefit from `near-misses' between classes, soft work sharing lets queries benefit from the traversal effort of their predecessors — effort that hard work sharing discards — broadening what information produced by search is worth reusing and generalizing naturally across indexes.}
}


\vspace*{-0.2cm}
\section{Vector Joins}\label{sec:vector_joins}


This section defines threshold-based vector join and introduces the basic execution models used throughout the paper. 


\vspace*{-0.2cm}
\subsection{Problem Definition}\label{subsec:vector_joins:problem}




\begin{definition2}
\label{def:vector_join}
\cite{SimJoin}
\textbf{(Vector Join)} 
Given two sets of vectors $X$ and $Y$ in $\mathbb{R}^d$, a distance function $\dist(\cdot,\cdot)$, and a distance threshold $\theta$, the vector join is to find the set of all similar vector pairs whose distances are within the threshold, i.e., $X \Join_{dist < \theta} Y = \{(x, y) | x \in X, y \in Y, dist(x, y) < \theta\}$.
\end{definition2}
\vspace*{-0.2cm}


We assume without loss of generality that $|X| \le |Y|$, and refer to $X$ as the \emph{query} set and $Y$ as the \emph{data} set. 
We often refer to a vector as a \emph{point} in vector space.

Unlike ad-hoc top-$k$ search, threshold-based vector join is typically executed over a predefined batch of query vectors, similar to a join key column in relational joins. This makes it natural to assume a predefined index on queries as in \cite{SimJoin}.
In practice, vector joins are often executed in batch mode over large query sets, such as offline analytics pipelines, periodic data cleaning jobs, or hybrid vector-relational processing \cite{DuckDBV, chen2024singlestore, zhang2023vbase}. Users can join two predefined vector sets, and each vector can have attributes corresponding to the structured metadata of vectors.

\moved{\textbf{Answer to: Isn't merged index used for a single join only? Wouldn't it be expensive to build a merged index for every pairwise join between $N > 2$ sets?} Repeated joins over a maintained pair of sets are the norm in this workload model, not the exception. In cross-source product matching (Section \ref{sec:intro}, Appendix A), for example, a marketplace joins a maintained catalog $Y$ against incoming seller batches $X$ under the same encoder, exercising three reuse axes: (i) the pair is re-joined on a schedule as data evolves, with the index maintained incrementally (Section \ref{subsec:ours:mi}); (ii) the join runs at many thresholds on the same data---a tight $\theta$ to auto-accept high-confidence pairs, and a looser $\theta$ to generate more candidates for a reranker or human review---reusing the same \emph{threshold-agnostic} index without rebuilding; and (iii) the same index also serves ordinary top-$k$ search on either set (Section \ref{subsec:exp:revision_single_set}). For multiple same-encoder sets, one union index with per-vector set labels serves all pairwise joins, avoiding $N(N{-}1)/2$ pairwise structures. The one-time offline indexing cost is thus amortized over many joins and searches, not charged to a single query.}

Threshold-based joins are also more sensitive than top-$k$ search to local density and data skew: dense regions may contain many in-range points, while sparse regions may contain very few. Fixed-$k$ retrieval does not reflect this intrinsic structure well, since it may miss relevant matches in dense regions or return unnecessary candidates in sparse ones. This requires additional search with increased $k$ or discarding some results \cite{SimJoin, zhang2023vbase}. Threshold-based joins, in contrast, directly capture semantic proximity under a consistent decision rule.


\vspace*{-0.1cm}
\begin{definition2}
\textbf{(Approximate Vector Join)} 
Approximate vector join is to find the threshold-qualified vector pairs in an approximate way to achieve higher efficiency.
\end{definition2}
\vspace*{-0.1cm}

As in ANN search, we measure the quality of an approximate vector join by \emph{recall}, the fraction of returned join pairs relative to the exact result $X \Join_{\dist<\theta} Y$. In the paper, we focus on graph-based ANN indexes \cite{HNSW, NSG} since they offer the best efficiency--recall trade-off among existing vector-search families and are used by the current state of the art in approximate vector join \cite{SimJoin}.

\vspace*{-0.1cm}
\begin{definition2}
\label{def:graph}
\textbf{(Graph-based Index and Search)} 
A graph-based index $G$ is a triple $(V,E,s)$ where $V$ is the set of nodes representing vectors, $E$ is the set of directed edges between nodes, forming a graph, and $s$ is the designated entry point for search.
\end{definition2}
\vspace*{-0.1cm}

We use $(u,v)$ to denote a directed edge from node $u$ to node $v$, and $v$ is an (outgoing) \emph{neighbor} of $u$. 
Given a query $x$, graph-based search starts from $s$ and traverses the graph toward points close to $x$, computing query-specific distances along the way. The underlying graph $(V, E)$ is a proximity graph.

In this paper, we use $G_A$ to denote a pre-built graph index over vector set $A$. Depending on the execution mode, the framework may use separate indexes for the query and data sets, or a single merged index over both. Regardless of that choice, search is still driven by local graph traversal and repeated query-specific distance computations.

\vspace*{-0.2cm}
\subsection{Basic Execution Models}\label{subsec:vector_joins:execution_models}

This section summarizes the basic execution models for vector join.


\subsubsection{Nested Loop Join (NLJ)} 
The naive baseline computes all pairwise distances between $X$ and $Y$, with complexity $O(|X||Y|d)$ for $d$-dimensional vectors. It produces the exact join result, but becomes prohibitive for large datasets and high-dimensional embeddings.

\subsubsection{Index Nested Loop Join (INLJ)} 
Each query $x\in X$ independently probes an index $G_Y$ on the data $Y$. \sigmod{Once an in-range region is found via greedy search (navigating closer to the query), it traverses within that region via BFS \cite{SimJoin}.}
This is much faster than NLJ, but it still repeats graph traversal from the fixed entry $s_Y$ for every query, even when many queries are similar. This repeated online work is the main motivation for work sharing.

\subsubsection{INLJ with Work Sharing (WS)}
The state of the art augments INLJ with work sharing by processing similar queries in a coordinated order and reusing information from one query for another \cite{SimJoin}. It first constructs an MST over query graph $G_X$, which takes negligible time during vector join, and uses this MST to order queries starting from its root. Once a parent query completes, its approximate join results are cached and used as traversal seeds for its child queries.

The benefit of reusing the parent $p$'s join results for its child $x$ is approximately quantified as their distance $dist(x, p)$, from the insight that more similar queries (smaller distance) have more similar join results (larger benefits). Since the MST over the queries minimizes the sum of such inter-query distances, it maximizes the total benefits among possible query orderings. Moreover, when building this MST, not only the query-graph edges but every $(s_Y, x)$ for each $x \in X$ is used in order to 1) ensure the connectivity of MST and 2) let a query fall back to start from the entry $s_Y$ if it is closer to the query than the parent. $s_Y$ is regarded as the root of MST.

\vspace*{-0.2cm}
\subsection{Why Not Hash-Join-Like Designs?}\label{subsec:vector_joins:why_not_hash_join}

Although query ordering and work sharing reduce traversal overhead, vector join remains fundamentally different from classical relational hash join because distance computation is query-specific. Even if multiple queries reach the same point $y$, each query $x$ must still compute $\dist(x,y)$ and compare it with the threshold $\theta$. This repeated high-dimensional distance evaluation is the main reason why vector joins have so far been dominated by nested-loop-like or index-nested-loop-like designs.

At the same time, this observation also suggests what is missing (Table \ref{tab:rdbms_analogy}). Existing methods still spend online effort to discover an initial in-range region for each query. A more aggressive design would shift part of this effort to the build phase, so that runtime search can begin from a much more relevant region, hopefully achieving \emph{constant-time lookup, i.e., in-range-region discovery}. Our merged index moves toward this design point by encoding cross-set geometry during index construction. 

\begin{table}[t]
\centering
\small
\setlength{\tabcolsep}{4pt}
\caption{\revisionA{Novelty}{Relational analogy for vector joins.}}
\label{tab:rdbms_analogy}
\vspace*{-0.2cm}
{
\begin{tabular}{lccc}
\toprule
& \multicolumn{2}{c}{\textbf{Relational join cost}} & \textbf{Has vector-join} \\
\cmidrule(lr){2-3}
\textbf{Join type} & \textbf{Build} & \textbf{Probe} & \textbf{counterpart?} \\
\midrule
Nested loop join       & --               & High    & \cmark \\
Index nested loop join & High (offline)   & Medium & \cmark \\
Hash join              & Medium (online)  & Low    & \xmark \\
\bottomrule
\vspace*{-0.2cm}
\end{tabular}
}
\vspace*{-0.5cm}
\end{table}

\vspace*{-0.2cm}
\section{Fast Approximate Vector Joins via Offline--Online Co-Design}\label{sec:ours}



\revisionA{}{

We now present our offline--online co-design for fast approximate vector join. The key idea is to address three orthogonal issues identified in Section~\ref{sec:intro}: \emph{redundancy across similar queries}, \emph{redundancy across joins} that repeatedly rediscovers cross-set geometry, and \emph{reachability failure on difficult queries} whose valid matches lie in disconnected in-range regions. To do so, we combine three components in one unified execution pipeline: \emph{soft work sharing} for online reuse, an \emph{encoder-centric merged index} for offline work offloading, and an \emph{OOD-aware adaptive hybrid search} for the remaining hard cases. 
These components are complementary \sigmod{and form a coherent design: the merged index is the structural foundation that eliminates redundant cross-set geometry discovery and provides the OOD signal needed for adaptive hybrid search; soft work sharing then maximizes reuse of whatever online traversal remains; and hybrid search fills the residual recall gap that neither offline preparation nor reuse can resolve.}
\nata{This stands in contrast to prior work such as \cite{SimJoin}, which optimizes the online search phase while assuming a fixed, table-centric index organization; our key departure is to change the index itself — specifically to empower the join.}
}

\vspace*{-0.2cm}
\subsection{Unified Execution Framework}\label{subsec:ours:framework}

Algorithm~\ref{alg:vector_join} shows the end-to-end framework. Given query vectors $X$, data vectors $Y$, graph indexes, threshold $\theta$, and configuration parameters \texttt{cfg} (Table \ref{tab:cfg}), the framework first loads either separate indexes $\{G_X,G_Y\}$ or a merged index $G_{X\cup Y}$, determines a query processing order, and initializes a query-level cache. It then processes each query in four steps: (1) construct a seed set, (2) perform greedy search to reach an initial in-range point, (3) determine whether hybrid expansion is needed, and (4) perform post-greedy expansion to collect the remaining in-range points (Figure \ref{fig:search_phases} illustrates this search procedure). Finally, it updates the cache according to the selected work-sharing mode. This modular structure makes the algorithmic differences between separate vs.\ merged indexing, hard vs.\ soft work sharing, and BFS vs.\ hybrid expansion explicit within one common pipeline. 
\nata{Figure \ref{fig:positioning} maps this design space; \cite{SimJoin} occupies the dashed lines (separate indexes, hard work sharing, BFS), while our full system extends along all three dimensions (shown in solid lines).}
\footnote{\moved{\textbf{Answer to: Why not used the original implementation of \SimJoin?}} We referred to the algorithms in \cite{SimJoin} to implement their approach within our framework but with several modifications for clarity and correctness. 
For example, their algorithms failed to track already-visited data points, which caused the same distance computation to be performed multiple times.
As their code was not publicly available, we reimplemented their method within our framework; we achieved comparable latency (around 1ms per query on average) on the datasets we used and under small thresholds, and speedup when using their hard work sharing.
}

\begin{algorithm}[t]
\caption{\revisionA{}{\VectorJoin}}\label{alg:vector_join}
\small
\KwIn{Query vectors $X$, data vectors $Y$, index retriever $G$, threshold $\theta$, configuration \texttt{cfg}}
\KwOut{Join result set $R$}
    Load $\{G_X,G_Y\}$ or $G_{X\cup Y}$ according to \texttt{cfg.merged} \\
    $O \leftarrow \textsc{OrderQueries}(G,\texttt{cfg})$ \\
    Initialize query-level cache $C$, result set $R\leftarrow \emptyset$ \\
    \ForEach{$q\in O$}{
      $S \leftarrow \textsc{SeedPhase}(q,C,G,\texttt{cfg})$ \\
      $(I,V)\leftarrow \textsc{GreedySearchPhase}(q,S,Y,G,\theta,\texttt{cfg})$ \\
      $\mathit{hybrid}\leftarrow \textsc{CheckHybrid}(q,G,\texttt{cfg})$ \\
      $I \leftarrow \textsc{ExpansionPhase}(q,\mathit{hybrid},I,V,Y,G,\theta,\texttt{cfg})$ \\
      $R \leftarrow R \cup I$; $\textsc{UpdateCache}(q,C,I,V,\texttt{cfg})$ \\
    }
    \Return{$R$}
\end{algorithm}

\begin{table}[t]
\centering
\small
\setlength{\tabcolsep}{4pt}
\caption{\revisionA{}{Configuration parameters in \texttt{cfg}.}}
\label{tab:cfg}
\vspace*{-0.2cm}
\begin{tabular}{p{0.26\columnwidth}p{0.70\columnwidth}}
\toprule
\textbf{Field} & \textbf{Meaning} \\
\midrule
\texttt{merged} & Whether the index is merged or not. \\
\texttt{work\_sharing} & Work-sharing mode: \texttt{none}/\texttt{hard}/\texttt{soft}/\texttt{self}. \\ 
\texttt{force\_hybrid} & Always enable hybrid search. \\
\texttt{adapt\_hybrid} & Enable hybrid search only when the query is predicted OOD. \\
\texttt{greedy\_beam\_size} & Max queue size in greedy search (default: 256). \\
\texttt{conv\_threshold} & Threshold for determining convergence in greedy search (default: 10, \sigmod{sensitivity analysis in Appendix D due to space limitation}). \\
\texttt{hybrid\_beam\_size} & Max queue size in hybrid search for out-range points (default: 256, \sigmod{sensitivity analysis in Figure \ref{fig:latency_recall_exp}}). \\
\bottomrule
\end{tabular}
\end{table}

\begin{figure}[h!]
\centering
\includegraphics[width=1.0\columnwidth]{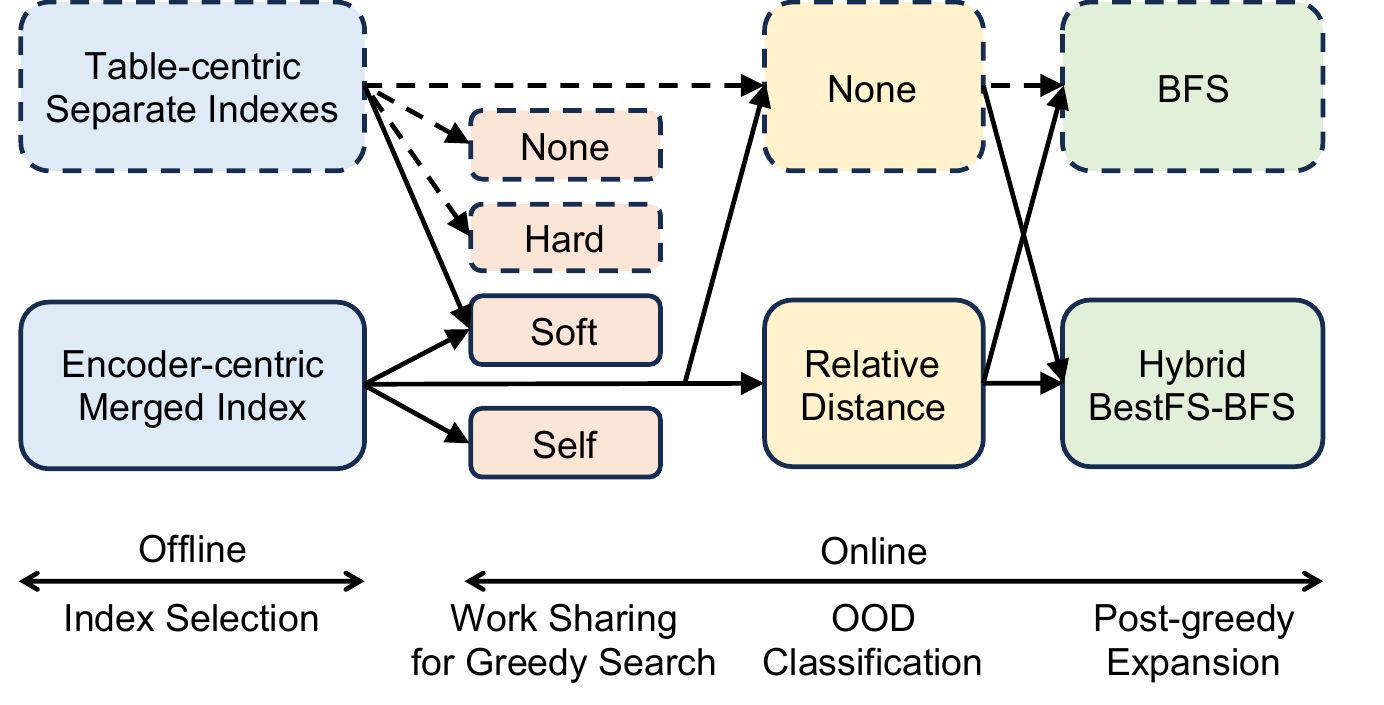}
\vspace*{-0.6cm}
\caption{\revisionA{Novelty}{Design space of approx. vector join. Dashed lines are explored in existing work, whereas solid lines are our contributions. 
OOD classification is tied to merged index.}}\label{fig:positioning}
\end{figure}

\begin{figure}[h!]
\centering
\includegraphics[width=0.8\columnwidth]{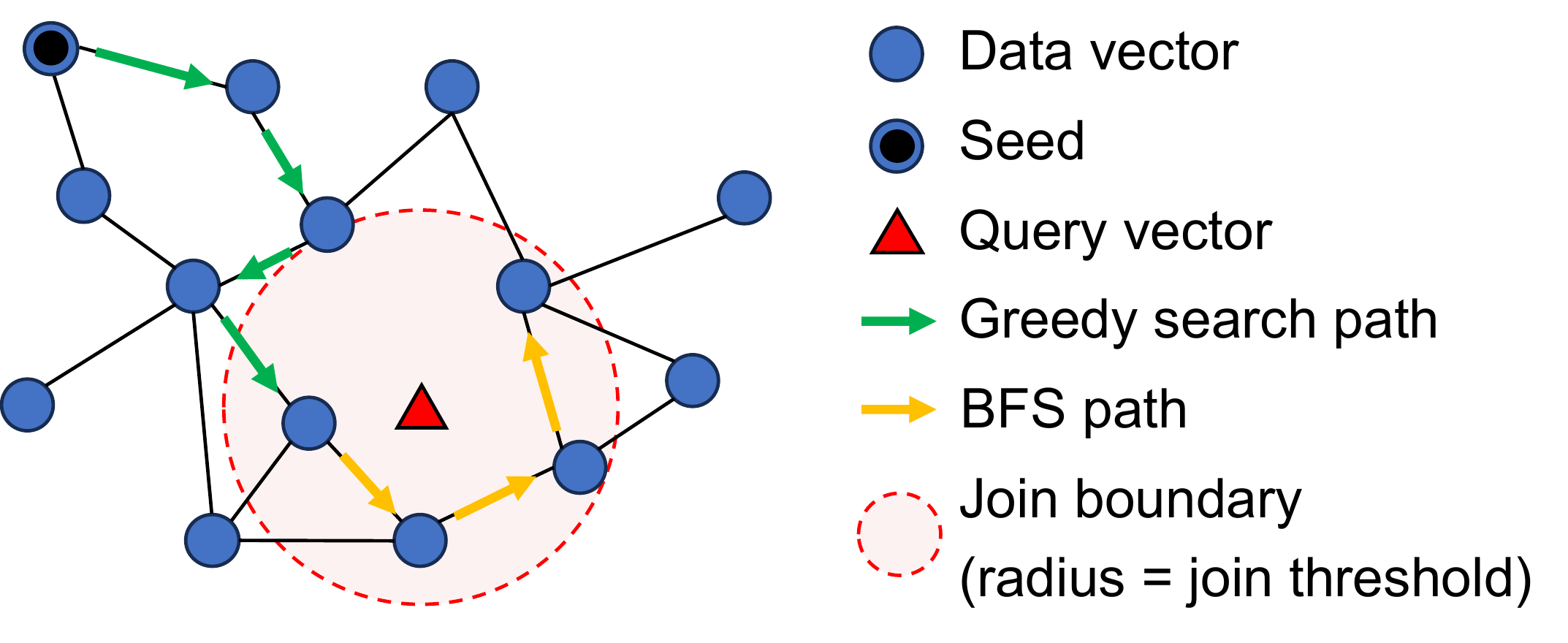}
\vspace*{-0.4cm}
\caption{\revisionA{}{Search procedure for a single query: starts from a seed set, performs greedy search to reach an initial in-range point, and then performs post-greedy expansion (BFS here) to collect remaining in-range points.}}\label{fig:search_phases}
\end{figure}

\subsubsection{Ordering and seed selection}


The quality of the seed is critical. The closer the starting node is to the query, the fewer search steps are typically needed to reach an in-range region.
In \textsc{OrderQueries} (Algorithm \ref{alg:order_queries}), when work sharing is enabled in hard or soft mode, queries are processed in an MST-based order so that similar queries are adjacent; otherwise they may be processed in arbitrary order. 
The \textsc{SeedPhase} (Algorithm \ref{alg:seed_phase}) then constructs the seed set. The default seed is the global entry point of the relevant index.
Depending on \texttt{cfg.work\_sharing}, this seed set may be augmented with cached seeds from the parent query (hard or soft mode) or with the query point itself (self mode) for merged index. This abstraction lets separate-index and merged-index execution share the same outer pipeline while differing only in how seeds are chosen.

\begin{algorithm}[htb]
\caption{\revisionA{}{\textsc{OrderQueries}}}\label{alg:order_queries}
\small{
\KwIn{$G$, \texttt{cfg}}
\KwOut{Order of queries}
    \If{\texttt{cfg.work\_sharing} $\in \{ \mathrm{hard, soft} \}$}{
        \Return the MST-based order from $G_X$ or $G_{X \cup Y}$ according to \texttt{cfg.merged}
    }
    \Return an arbitrary order
}
\end{algorithm}
\vspace*{-0.2cm}

\begin{algorithm}[htb]
\caption{\revisionA{}{\textsc{SeedPhase}}}\label{alg:seed_phase}
\small{
\KwIn{Query $q$, cache $C$, $G$, $\texttt{cfg}$}
\KwOut{Seed set $S$}
    $S \leftarrow \{s_{X \cup Y}\}$ \textbf{if} \texttt{cfg.merged} \textbf{else} $\{s_Y\}$ \redcomment{from $G_{X \cup Y}$ or $G_Y$}
    \Switch{\texttt{cfg.work\_sharing}}{
      \Case{hard \textbf{or} soft}{ 
        $S \leftarrow S \cup \mathcal{C}[q.parent]$ \redcomment{hard/soft work sharing}
      }
      \Case{self}{ 
        $S \leftarrow S \cup \{q\}$ \redcomment{merged index}
      }
    }
    \Return $S$
}
\end{algorithm}

\subsubsection{Greedy search}


As shown in Figure~\ref{fig:search_phases}, the search begins with a greedy search (Algorithm \ref{alg:greedy_phase}), implemented as a best-first traversal over a min-priority queue ordered by distance to the query. Starting from the seeds, the algorithm explores progressively closer nodes until one of three stopping conditions is met (Algorithm \ref{alg:stop_greedy}): (1) an in-range data point is found, (2) the best queue distance has not improved for a fixed number of iterations, or (3) in the merged-index case, the search has already probed the neighborhood of the query node that is likely to have the closest data point.
Here, (2) is a simple early-stopping rule that avoids wasteful traversals when no in-range point is likely to be found, which is especially important at small thresholds. \sigmod{We detail (3) in Section \ref{subsec:ours:mi}.}
The greedy search returns both the in-range points found so far and the set of visited nodes to share with post-greedy expansion, so no duplicate distance computations occur.

\moved{\textbf{Answer to: Why does a larger queue size not improve recall?}} This greedy search resembles standard top-$k$ search \cite{HNSW, NSG}, but its goal is different: it only needs to find one in-range point rather than rank the $k$ closest points. In top-$k$ search, the search ends if the queue is stabilized, thus larger queues (at least size of $k$) can improve accuracy at the cost of more traversal. However, the queue size here matters much less, as the task is closer to top-1 search.


\begin{algorithm}[htb]
\caption{\revisionA{}{\textsc{GreedySearchPhase}}}\label{alg:greedy_phase}
\small{
\KwIn{Query $q$, seeds $S$, $Y$, $G$, $\theta$, \texttt{cfg}}
\KwOut{Discovered in-range points $I$, visited points $V$}
    $I \leftarrow \{u | u \in S, dist(u, q) < \theta\}$ \redcomment{in-range data points discovered}
    $V \leftarrow S$ \redcomment{visited points}
    Initialize $Q$ from $S$ \redcomment{min-priority queue by ${dist}(\cdot, q)$}
    \While{$Q \neq \emptyset$ \textbf{and not} $\textsc{StopGreedy}(q, \mathrm{front}(Q), Y, \theta, \texttt{cfg})$}{
      $u \leftarrow \mathrm{popFront}(Q)$ \\
      \ForEach{neighbor $v$ of $u$ with $v \notin V$}{
          $V \leftarrow V \cup \{v\}$; $\delta \leftarrow \mathrm{dist}(v, q)$; insert $v$ into $Q$ with key $\delta$ \\
          \If{$v \in Y$ \textbf{and} $\delta < \theta$}{
              $I \leftarrow I \cup \{v\}$ \\
          }
      }
      trim $Q$ to at most \texttt{cfg.greedy\_beam\_size} smallest keys \\
    }
    \Return $(I, V)$ \\
}
\end{algorithm}

\begin{algorithm}[htb]
\caption{\revisionA{}{\textsc{StopGreedy}}}\label{alg:stop_greedy}
\small{
\KwIn{Query $q$, front element $(d,u)$ of greedy queue $Q$, $Y$, $\theta$, \texttt{cfg}}
\KwOut{Whether to stop greedy walk}
    \textbf{if} $(d < \theta)$ \textbf{and} $(u \in Y)$ \textbf{then return true} \\
    \textbf{if} $d$ has not been decreased for \texttt{cfg.conv\_threshold} iterations in \textsc{GreedySearchPhase} \textbf{then return true} \\
    \textbf{if} \texttt{cfg.merged} \textbf{and} $u \neq q$ \textbf{then return true} \\
    \Return \textbf{false}
}
\end{algorithm}

\subsubsection{Post-greedy expansion}
\nata{The post-greedy expansion (Algorithm \ref{alg:expansion_phase}) collects remaining in-range points. By default it runs BFS over in-range nodes; when hybrid search is enabled, out-range nodes are also queued and the queue is ordered by distance rather than FIFO, so in-range regions are exhausted before out-range bridges are crossed. \textsc{UpdateCache} is explained in Section \ref{subsec:ours:sws}.}


\begin{algorithm}[htb]
\caption{\revisionA{}{\textsc{CheckHybrid}}}\label{alg:check_hybrid}
\small{
\KwIn{$q$, $G$, \texttt{cfg}}
\KwOut{Whether to use hybrid search}
  \uIf{\texttt{cfg.merged} \textbf{and} \texttt{cfg.adapt\_hybrid}}{
      \Return (whether query $q$ is predicted OOD)
  }
  \Else{
    \Return \texttt{cfg.force\_hybrid}
  }
}
\end{algorithm}

\begin{algorithm}[htb]
\caption{\revisionA{}{\textsc{ExpansionPhase}}}\label{alg:expansion_phase}
\small{
\KwIn{Query $q$, hybrid search enabled $hybrid$, initial in-range points $I$, visited $V$, $Y$, $G$, $\theta$, \texttt{cfg}}
\KwOut{Final in-range points $I$}
    Initialize $Q$ from $I$ \redcomment{min-priority queue by ${dist}(\cdot, q)$ if \texttt{cfg.merged}, otherwise FIFO}
    \While{$Q \neq \emptyset$}{
      $u \leftarrow \mathrm{popFront}(Q)$ \\
      \ForEach{neighbor $v$ of $u$ with $v \notin V$}{

          \If{\texttt{cfg.merged} \textbf{and} $v \notin Y$}{
              \textbf{continue} \\
          }
          $V \leftarrow V \cup \{v\}$; $\delta \leftarrow {dist}(v, q)$ \\
          \If{$\delta < \theta$}{
              insert $v$ into $Q$; $I \leftarrow I \cup \{v\}$ \\
          }
          \Else{
              \uIf{$hybrid$}{
                  insert $v$ into $Q$; trim $Q$ to at most max(\texttt{cfg.hybrid\_beam\_size}, $|Q| - 1$) keys \redcomment{don't discard in-range points}
              }
          }
      }
    }
    \Return $I$ \\
}
\end{algorithm}

\vspace*{-0.2cm}
\subsection{Soft Work Sharing (Online)}\label{subsec:ours:sws}

\sigmod{We describe our online reuse first as it is useful independently of the index organization.}
In Section~\ref{subsec:ours:framework}, the online reuse corresponds to using \texttt{cfg.work\_sharing = hard} or \texttt{soft}. The only algorithmic difference between the two is the cache-update policy,  \textsc{UpdateCache} (Algorithm \ref{alg:update_cache}).

\begin{algorithm}[htb]
\caption{\revisionA{}{\textsc{UpdateCache}}}\label{alg:update_cache}
\small{
\KwIn{Query $q$, cache $C$, in-range points $I$, visited $V$, \texttt{cfg}}
    \Switch{\texttt{cfg.work\_sharing}}{
      \Case{hard}{
        $\mathcal{C}[q] \leftarrow I$ \\
      }
      \Case{soft}{
        $\mathcal{C}[q] \leftarrow $ top-1 closest data point in $V$ to $q$ \\
      }
    }
}
\end{algorithm}

Hard work sharing (HWS), used by prior work \cite{SimJoin}, caches all in-range points found for the current query and reuses them as seeds for its children in the MST. This is effective when nearby queries already share many answers, but it is restrictive in two ways. First, when the threshold is small, the search may find no in-range point at all, so the cache is empty and all traversal effort is lost. Second, when the threshold is large, caching all in-range points is redundant, because the next query only needs a good seed to reach an in-range region, not the full previous answer set. 

Our soft work sharing (SWS) addresses these by storing only the top-1 closest visited point, even if that point is out of range for the current query. In terms of the framework, this means setting \texttt{cfg.work\_sharing = soft} and updating \textsc{UpdateCache} to store the best point in \texttt{visited}, rather than the set of in-range points. The key intuition is that search produces reusable information beyond final join answers: even without finding an in-range point, the closest discovered point may still be much closer to a similar future query than the global entry node of the index. Thus, SWS reuses \emph{useful search effort} rather than only \emph{search results}, while also keeping the cache compact. 
Importantly, this benefit does not depend on merged indexing; SWS is already useful under separate-index executions because it directly reduces redundancy across similar queries, mainly reducing greedy-search steps and improving efficiency without hurting recall compared to HWS.

\vspace*{-0.2cm}
\subsection{Merged Index (Offline)}\label{subsec:ours:mi}

We next move beyond improving reuse under separate indexes and reduce the online effort more fundamentally by changing the index organization itself. In the framework, this corresponds to setting \texttt{cfg.merged = true}, using \texttt{cfg.work\_sharing = self} by default. 

\revisionA{R2.D2}{The key observation is that vectors to be joined are typically produced by the same encoder \cite{pansurvey} and already inhabit a shared embedding space. 
\nata{Building separate indexes forces the online phase to rediscover cross-set proximity from scratch for every join — a cost that can be offloaded to index construction.}
A more natural design — one that is encoder-centric — directly organizes vectors from both sets in one unified graph, encoding cross-set geometry into the index structure itself.}

\sigmod{The question is then how to exploit this unified structure to accelerate the initial online in-range region discovery. A naive approach would be to precompute in-range regions for each query during index construction. However, this faces two fundamental obstacles: the threshold $\theta$ is not known at construction time, and even if it were fixed, a separate structure per threshold would be needed, incurring prohibitive overhead. We therefore seek a threshold-agnostic solution.}

\sigmod{The key insight is that such a solution is already implicit in standard graph indexes, through the \emph{relative neighborhood graph} (RNG) property \cite{RNG, NSG}. 
We \emph{revisit} this property in the context of threshold-based vector join, showing that it is sufficient for threshold-agnostic in-range region discovery without building a dedicated index.
An edge $(u, v)$ exists in an RNG when no other point is simultaneously closer to both $u$ and $v$, preserving local proximity while avoiding redundant edges. A direct consequence proved formally in Appendix C is that the top-1 closest node to any query lies in the query's local neighborhood. Therefore, once query and data vectors are organized in a single unified graph, a query node's neighborhood highly likely exposes its nearest data-side neighbor without any threshold-specific precomputation, enabling constant-time in-range region discovery}: $O(md)$ for $m$ neighbors and vector dimension $d$. This termination is captured by \textsc{StopGreedy} (Line 3). 
The merged index (MI) implements this idea for any off-the-shelf graph index by building $G_{X \cup Y}$ instead of $G_X$ and $G_Y$ (Figure \ref{fig:separate_merged_index}). In this sense, the MI is analogous to a hash-join-like design point for vectors (constant-time lookup) while the build phase is offline as in index nested-loop join (Table \ref{tab:rdbms_analogy}). \nata{If multiple vector sets are produced by the same encoder, they can be indexed jointly in a single structure; in this paper, we focus on the common case of two sets, $X$ and $Y$.}



\begin{figure}[ht]
\vspace*{-0.2cm}
    \centering
    \subfigure[Separate indexes.]{%
        \includegraphics[width=0.33\columnwidth]{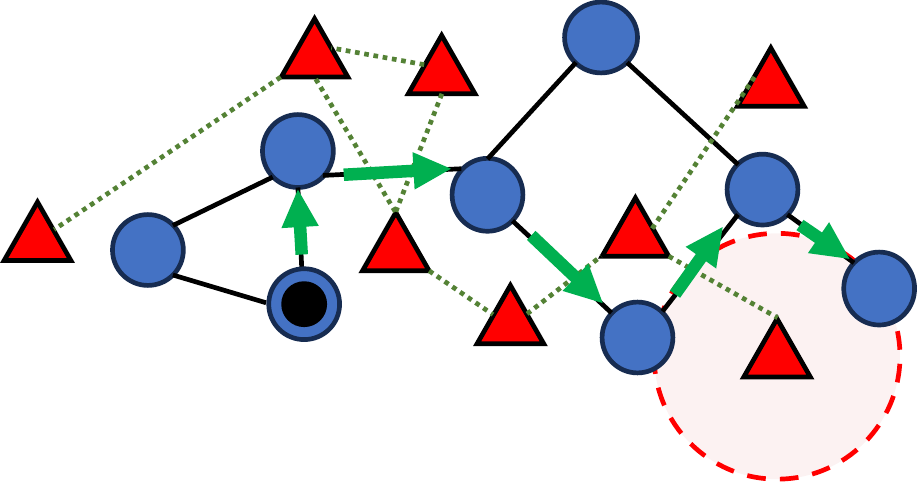}
        \label{fig:separate_index}
    }
    \hfill
    \subfigure[Merged index.]{%
        \includegraphics[width=0.33\columnwidth]{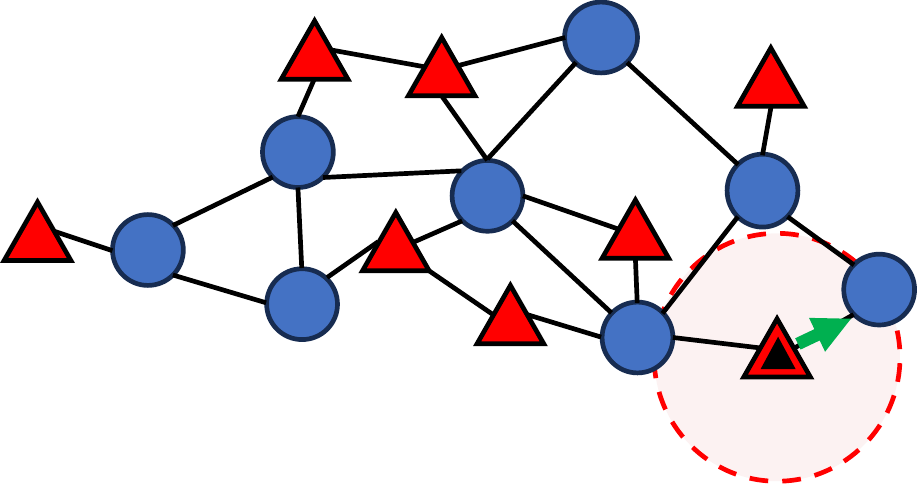}
        \label{fig:merged_index}
    }
    \vspace*{-0.4cm}
    \caption{Search procedure for two types of indexing. Same legend with Figure \ref{fig:search_phases}. Bottom-right query is used. 
    }
    \label{fig:separate_merged_index}
\vspace*{-0.2cm}
\end{figure}

\revisionA{R3.W3}{A caveat is that a query may have no neighboring data point. However, such cases are rare in practice: even on OOD datasets such as \ImageNet{} (Section \ref{sec:exp}), where query and data distributions barely overlap, only \sigmod{up to} 1\% of queries exhibit this behavior, mostly with empty join results due to the large distance to data. Hence, we leave adding more specialized techniques \sigmod{(e.g., forcing cross-set neighbors)} on top of MI as future work.}

The offline overheads of MI are at most modest. Since graph-based indexes store a bounded number of outgoing neighbors per node, the MI preserves the same asymptotic linear space complexity as maintaining separate indexes on $X$ and $Y$, and its construction time remains in the same asymptotic order. Thus, the MI is not a heavier asymptotic object; it reorganizes essentially the same indexing budget into a unified structure that enables offline work offloading. \revisionA{R1.D5}{We prove this theoretically in Appendix B due to space limitation, and show in Section \ref{sec:exp} that MI adds negligible overheads (< 5\%) in practice.}
Because the MI is built over the shared embedding space rather than for one specific threshold, it can be reused for arbitrary future threshold-based joins over the same shared-encoder vector sets.
Finally, the maintenance is also simple as we can use existing index update mechanisms \cite{HNSW, singh2021freshdiskann}.

\moved{\textbf{Answer to: How is the index maintained?} \textit{Discussion on index maintenance.} Since the MI is index-agnostic, existing index maintenance mechanisms apply directly \cite{HNSW, singh2021freshdiskann}, and maintenance optimization is largely orthogonal to this paper. A node update requires no query/data distinction---only its set label, which is consulted when the node is used in joins or searches. Incremental insertions handle streaming arrivals (e.g., seller batches in Appendix A), while periodic batch rebuilds can restore graph quality under heavy churn, exactly as for ordinary per-set indexes. The build cost itself is the standard graph-ANN construction cost one already pays for search, not a join-specific surcharge: building the MI costs essentially the same as the two per-set indexes it replaces (up to $+4.5\%$ size and $+0.7\%$ time, Section \ref{subsec:exp:offline_overhead}).}

\textit{Discussion on filtered queries.} The MI also resolves a structural limitation of work sharing over separate indexes; work sharing depends on the connectivity of the query-side graph and its induced MST. When the query set is filtered, the induced query subgraph can be highly disconnected, and no work can be shared across disconnected components. The MI avoids this dependence on query-side connectivity altogether, because the execution is anchored in the shared query-data geometry rather than only in the query-side structure.

\vspace*{-0.2cm}
\subsection{OOD-Aware Hybrid Search (Online)}\label{subsec:ours:ood_hybrid}



The previous components mainly target efficiency, but an important recall failure mode remains. This is caused by a strong \emph{locality assumption} shared by existing graph indexes and their search procedures. Once search reaches an in-range region, the remaining matches are expected to be reachable by continuing to expand locally through other in-range points. This assumption usually holds for easy in-distribution (ID) queries (e.g., \SIFT in Figure \ref{fig:tsne}), which is why BFS-style post-greedy expansion works well in prior work \cite{SimJoin}. However, it can fail for difficult OOD queries, where valid matches may lie in multiple disconnected in-range regions separated by out-range points (\LAION in Figure \ref{fig:tsne}). In such cases, BFS becomes trapped inside the first discovered in-range region and misses others.
Recovering the missing matches therefore requires selectively traversing out-range points as bridges.

\begin{figure}[htbp]
\vspace*{-0.2cm}
  \centering
  \begin{minipage}{0.49\columnwidth}
    \centering
    \includegraphics[width=\linewidth]{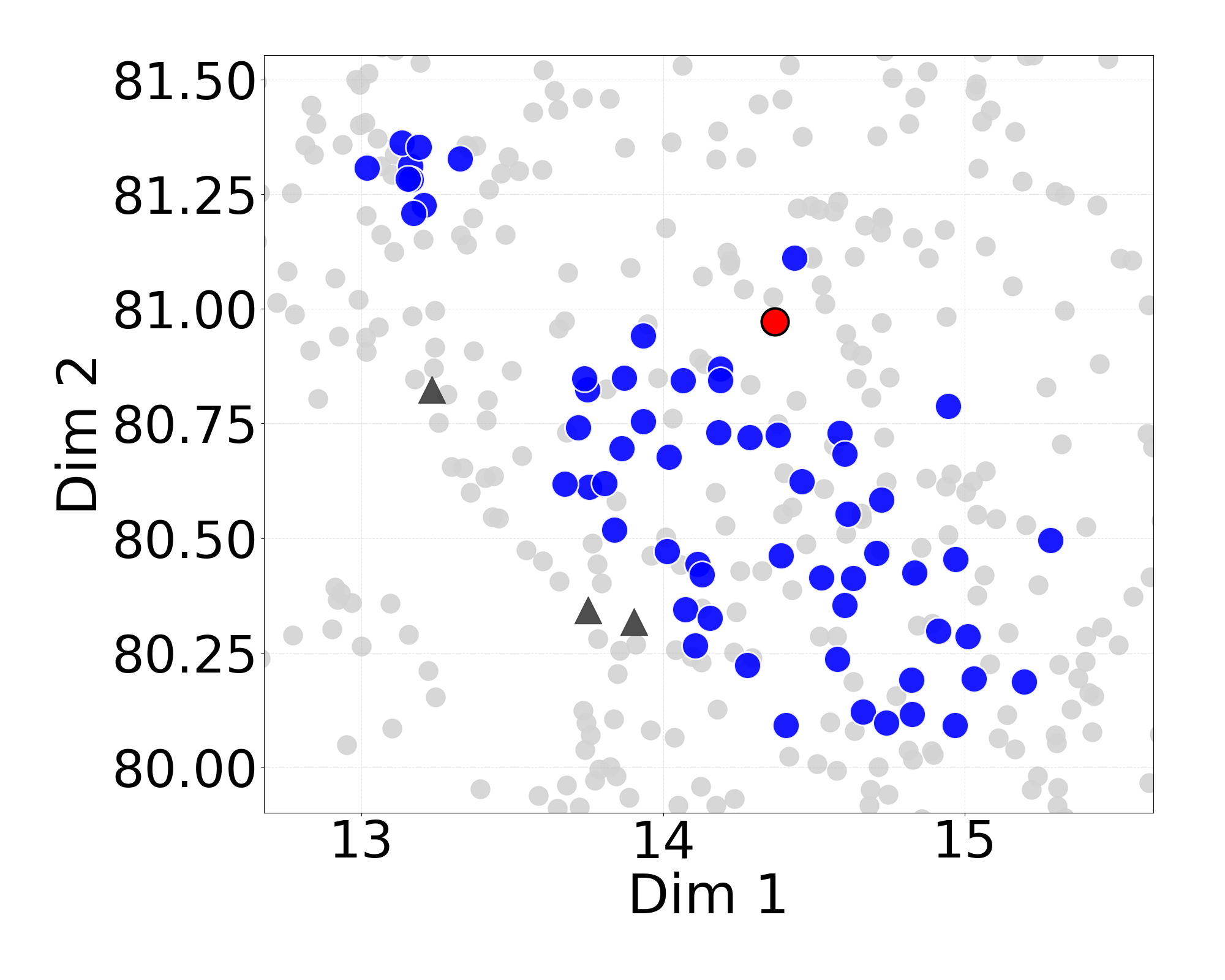}
  \end{minipage}
  \hfill 
  \begin{minipage}{0.49\columnwidth}
    \centering
    \includegraphics[width=\linewidth]{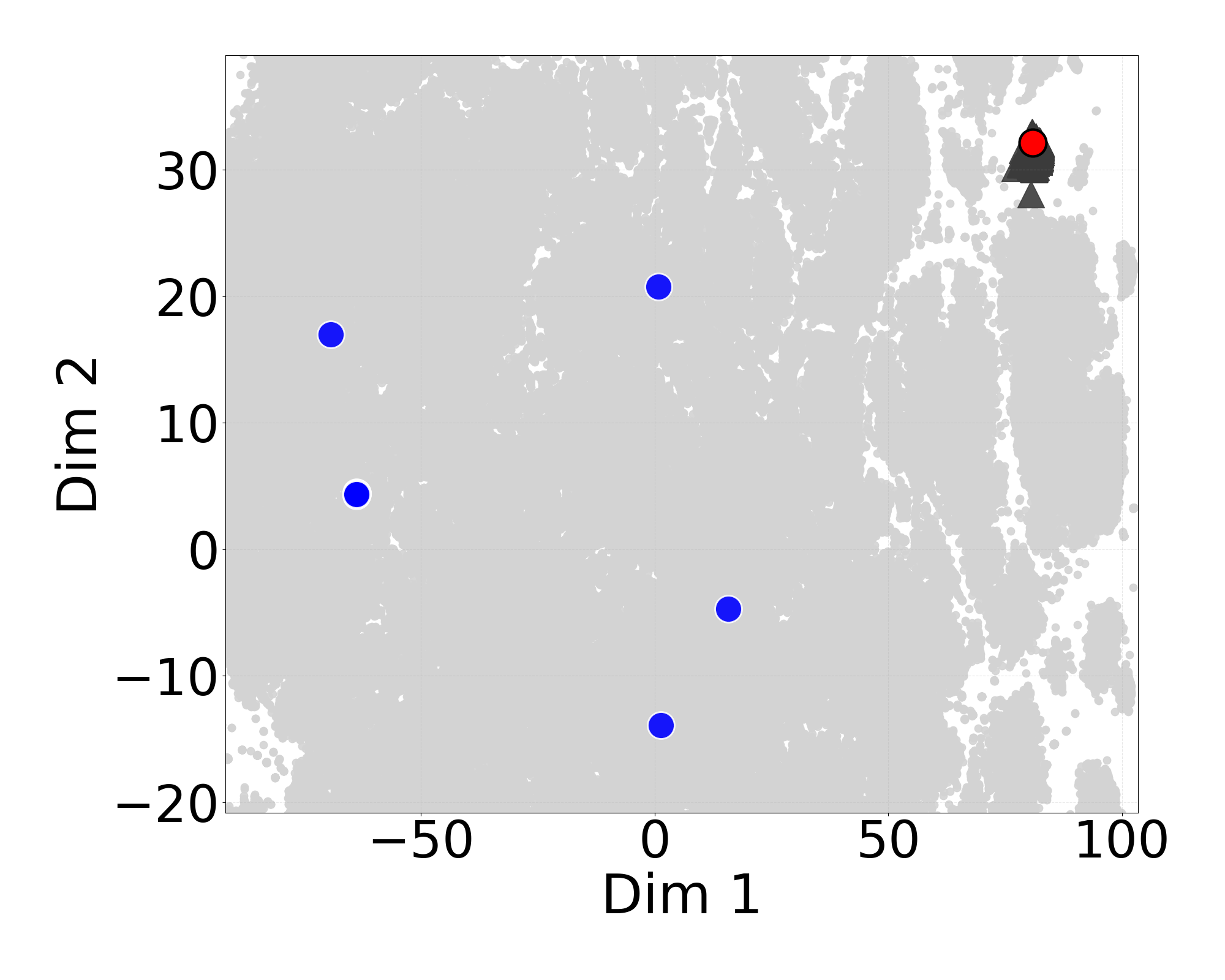}
  \end{minipage}
  \vspace*{-0.4cm}
  \caption{T-SNE visualization of \SIFT (left) and \LAION (right) datasets, highlighting an example query (red) and its in-range data points (blue).}
  \label{fig:tsne}
\vspace*{-0.4cm}
\end{figure}

\moved{\textbf{Answer to: Why not use the standard beam search of top-$k$ ANNS with a large queue?}} A natural alternative is best-first search (BestFS) \cite{NSG}, which always expands the closest point as in greedy search. 
However, BestFS is designed for top-$k$ retrieval rather than threshold search. It is effective for reaching close points, but it does not exploit the fact that, once an in-range region has been reached, all connected in-range points should be enumerated. Furthermore, as mentioned in Section \ref{sec:vector_joins}, BestFS requires $k$ as an input which is hard to predict for a given query and threshold \cite{lan2024cardinality}; Figure \ref{fig:sift_cdf_plot} shows that $k$ varies significantly across queries and thresholds even for the same dataset.
By contrast, BFS is well aligned with threshold search, but cannot traverse out-range points.
Our \emph{hybrid search} combines the strengths of both. Like BFS, it exhaustively expands currently discovered in-range points. Like BestFS, it also uses a bounded priority queue to keep out-range points that allows best-first exploration through them when necessary (Algorithm \ref{alg:expansion_phase}), \nata{but still prioritizing in-range regions.}
In this way, hybrid search preserves the completeness of BFS within each in-range region while borrowing the global reachability of BestFS across disconnected regions.

\begin{figure}[h!]
\vspace*{-0.2cm}
\centering
\includegraphics[width=0.54\columnwidth]{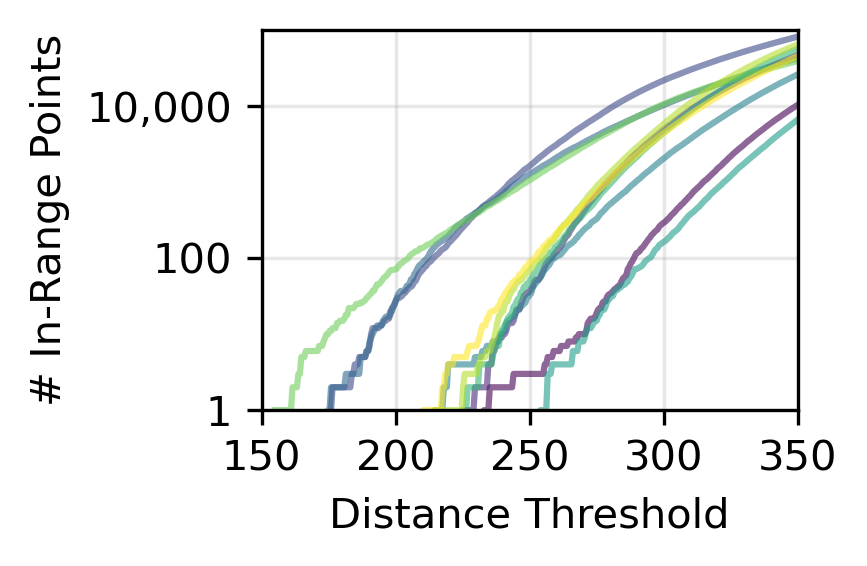}
\vspace*{-0.5cm}
\caption{Cumulative distributions of number of in-range data points, one line per sampled query, on \SIFT dataset.}\label{fig:sift_cdf_plot}
\vspace*{-0.2cm}
\end{figure}

The key remaining question is when to enable this extra exploration. Doing so would only add overhead for ID queries since the BFS is already sufficient. We therefore \emph{adaptively} enable hybrid search only for OOD queries. \revisionA{R1,2,3}{The merged index makes this strategy practical by providing a lightweight OOD prediction from its query--data overlap signal in the local neighborhoods (why OOD-aware hybrid search is tied to merged index in Figure~\ref{fig:positioning}).}
We measure 1) the average distance from a query to its neighboring data points ($d_1$ in Figure \ref{fig:OOD_combined}) {over} 2) the average distance from such neighboring data points to their neighbors that are 2-hop away from the query ($d_2$). If $\frac{d_1}{d_2} > 1.5$, we consider that this query lies far from data and is OOD. 
\revisionA{R1.W2}{
This metric is dimensionless and encoder-agnostic, and 1.5 is a robust threshold across the datasets evaluated (Figure \ref{fig:OOD_combined}).
}

\begin{figure}[h!]
\vspace*{-0.2cm}
\centering
\begin{minipage}[c]{0.33\columnwidth}
    \centering
    \includegraphics[width=\linewidth]{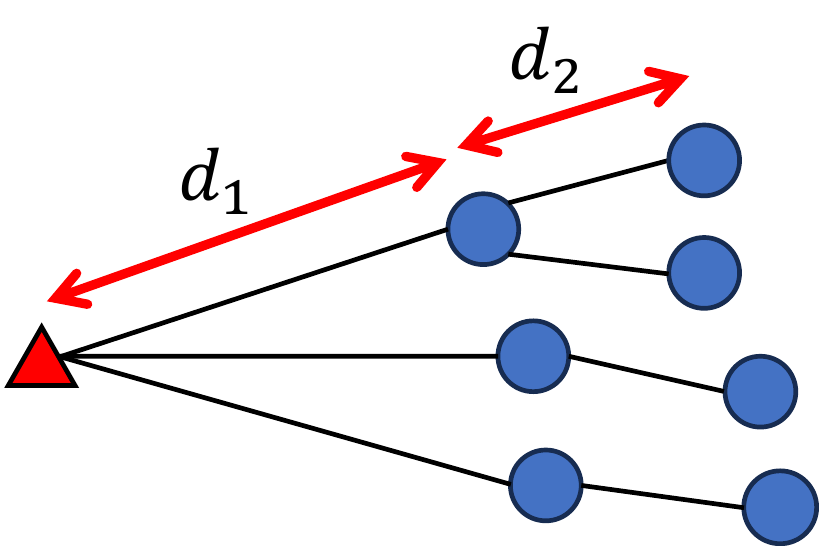}
\end{minipage}
\hfill
\begin{minipage}[c]{0.42\columnwidth}
    \centering
    \includegraphics[width=\linewidth]{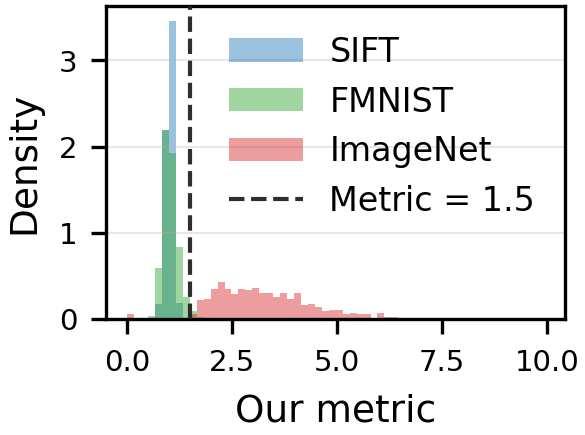}
\end{minipage}
\vspace*{-0.4cm}
\caption{(Left) Features used to predict whether a query is OOD. (Right) \revisionA{R1.W2}{Metric distributions for ID (\SIFT, \FMNIST) and OOD (\ImageNet) datasets.}}
\label{fig:OOD_combined}
\vspace*{-0.3cm}
\end{figure}

In the framework, this component corresponds to enabling either \texttt{cfg.force\_hybrid} or \texttt{cfg.adapt\_hybrid}. The latter is the default in our adaptive setting. \textsc{CheckHybrid} activates hybrid search only when the query is predicted OOD; otherwise, the framework uses the default BFS expansion.

\moved{\textbf{Answer to: Why do \MI and \MIA look similar in experiments?} Two behavioral consequences follow directly from this adaptivity. On in-distribution datasets, \MIA\ falls back to BFS for nearly all queries and therefore matches \MI's latency and recall; the two methods are \emph{expected} to coincide there. On OOD datasets (e.g., \ImageNet, \COCO, \LAION), \MIA\ deliberately spends extra out-range traversal to bridge disconnected in-range regions, trading a bounded amount of latency for the recall that all BFS-based methods lose (Section \ref{sec:exp}).} 


\paragraph{Remark.}
The three components address different failure modes and therefore complement one another. Soft work sharing reduces redundant traversal across similar queries and is beneficial even under separate indexes. The merged index improves the offline organization of the search space and largely removes the online effort of greedy search. OOD-aware hybrid search then fills the remaining recall gap for difficult queries whose valid matches are disconnected. In this sense, work sharing and hybrid exploration are complementary: work sharing is most beneficial for easier, more in-distribution queries, whereas hybrid exploration is needed mainly for harder OOD queries. 
\sigmod{Furthermore, when query and data distributions heavily overlap, combining soft work sharing on top of merged index can further improve recall, as the closest reused seed often falls within the join threshold for similar queries (Section \ref{subsec:exp:equi_size_join}).}

\section{Experiments}\label{sec:exp}

This section evaluates the performance of vector join methods implemented on our framework and answers the following questions. 


\begin{itemize}[topsep=0pt, leftmargin=0.5cm]
    \item Q1: Which \revisionA{}{join method} shows the best performance in terms of vector join execution time and recall?
    \item Q2: How do index parameters affect performance? 
    \item \revisionA{}{Q3: What is the main bottleneck?}
    \item Q4: How large is the overhead in offline index construction? 
    \item Q5: How well does a method scale to large data? 
    \item Q6: How does the proximity graph affect performance? 
    \item \moved{Q7: How does our method compare against a dedicated OOD index, and what is the trade-off for single-set searches?}
\end{itemize}

\subsection{Experimental Setup}\label{subsec:exp:setup}

\subsubsection{Datasets}

We evaluate our methods on eight widely used vector datasets (Table \ref{tab:datasets}) that cover diverse modalities, dimensionalities, and data distributions, including vision, text, and multimodal representations. Datasets are taken from the ANN-Benchmarks \cite{ann-benchmarks} and VIBE \cite{VIBE} suites, which provide standardized query-data splits and evaluation protocols for ANN search.


\begin{table}[htbp]
\vspace*{-0.2cm}
  \centering
  \caption{Statistics of datasets used (number of query vectors $|X|$, number of data vectors $|Y|$, dimension, mode of the degree distribution of data index $G_{Y}$, and the ratio of queries that are predicted as OOD by using our \revisionA{}{metric}). The last three rows are \revisionA{}{OOD datasets}.}
  \label{tab:datasets}
  \vspace*{-0.2cm}
  \scalebox{0.85}{
  \begin{tabular}{lccccc}
    \toprule
    \textbf{Dataset} & \textbf{$|X|$} & \textbf{$|Y|$} & \textbf{Dimension} & \textbf{Mode} & \textbf{OOD-Ratio} \\
    \midrule
    \SIFT     & 10,000    & 1,000,000 & 128 & 28 & 0.00\% \\
    \GIST     & 1,000     & 1,000,000 & 960 & 8  & 1.10\% \\
    \Glove    & 10,000    & 1,183,514 & 200 & 70 & 0.00\% \\
    \NYTimes  & 10,000    & 290,000   & 256 & 6  & 3.46\% \\
    \FMNIST   & 10,000    & 60,000    & 784 & 13 & 3.03\% \\
    \COCO     & 1,000     & 282,360   & 768 & 15 & 97.3\% \\
    \ImageNet & 1,000     & 1,281,167 & 640 & 19 & 97.4\% \\
    \LAION    & 1,000     & 1,000,448 & 512 & 22 & 95.1\% \\
    \bottomrule
  \end{tabular}
  }
\vspace*{-0.2cm}
\end{table}

\SIFT and \GIST consist of image descriptors commonly used for evaluating large-scale vector search systems.
\Glove contains word embeddings trained on large text corpora, capturing semantic similarity in natural language.
\NYTimes represents document embeddings derived from news articles, reflecting real-world text retrieval workloads.
\FMNIST consists of Fashion-MNIST image embeddings, representing dense visual feature vectors.
\COCO uses multimodal image-text embeddings from the MS-COCO dataset, capturing cross-modal semantic similarity.
\ImageNet contains deep visual embeddings extracted from the ImageNet dataset, representing large-scale image retrieval scenarios.
\LAION contains CLIP-based multimodal embeddings from the LAION dataset, reflecting modern web-scale vision-language representations. 
\kkim{Among these, \COCO, \ImageNet, and \LAION have mostly OOD queries \cite{VIBE}.
}

These datasets collectively span low to high dimensional spaces, small to million-scale data sizes, and both unimodal and multimodal embeddings, enabling a comprehensive evaluation of vector join behavior across different distributions.
For each dataset, we use the standard query-data splits provided by the benchmarks \revisionA{}{by default.}
In-distribution queries are spread across the data manifold rather than being clustered in a single region.

\moved{\textbf{Answer to: How were the distance thresholds designed?}} To evaluate threshold-based joins under varying selectivities, we use seven evenly spaced distance thresholds per dataset. The thresholds are chosen to cover a wide spectrum of join result sizes, from sparse joins with very few matches to dense joins where a large fraction of data vectors fall within the threshold. Table \ref{tab:epsilon_mapping} reports the threshold values used for each dataset. Figure \ref{fig:dataset_join_size} shows the 
join sizes, illustrating the high variance in join sizes across both datasets and thresholds.
\kkim{Note that large thresholds with large join sizes are less likely to be used in practice, similar to using small $k$ values (e.g., 5, 10, or 100) in typical top-$k$ ANN searches \cite{ootomo2024cagra, HNSW}.
Therefore, we mainly focus on small thresholds (e.g., $\theta_1$ to $\theta_3$) and show the results for larger thresholds for completeness.}

When the threshold is large, it is often enough to estimate the join size and approximate the results as in approximate query processing \cite{hellerstein1997online, agarwal2013blinkdb} or LIMIT the number of outputs as in exploratory data analytics \cite{idreos2015overview, zimmerer2025pruning}, and full sequential scans or nested loop joins may be more efficient than index scans or index nested loop joins, as the selectivity increases in the relational context \sigmod{(see Appendix E for further discussion)}. 





\begin{table}[ht]
\vspace*{-0.2cm}
\centering
\caption{Evenly-spaced distance thresholds ($\theta$) per dataset. 
}
\label{tab:epsilon_mapping}
\vspace*{-0.2cm}
\scalebox{0.85}{
\begin{tabular}{lccccccc}
\toprule
\textbf{Dataset} & \textbf{$\theta_1$} & \textbf{$\theta_2$} & \textbf{$\theta_3$} & \textbf{$\theta_4$} & \textbf{$\theta_5$} & \textbf{$\theta_6$} & \textbf{$\theta_7$} \\
\midrule
\SIFT      & 50  & 100 & 150  & 200  & 250  & 300  & 350  \\
\GIST      & 0.3   & 0.5   & 0.7    & 0.9    & 1.1    & 1.3    & 1.5    \\
\Glove     & 0.6   & 0.7   & 0.8    & 0.9    & 1.0    & 1.1    & 1.2    \\
\NYTimes   & 0.1   & 0.3   & 0.5    & 0.7    & 0.9    & 1.1    & 1.3    \\
\FMNIST    & 500 & 750 & 1000 & 1250 & 1500 & 1750 & 2000 \\
\COCO      & 1.33 & 1.335 & 1.34 & 1.345 & 1.35 & 1.355 & 1.36 \\
\ImageNet  & 1.19 & 1.21 & 1.23 & 1.25 & 1.27 & 1.29 & 1.31 \\
\LAION     & 1.12 & 1.14 & 1.16 & 1.18 & 1.2 & 1.22 & 1.24 \\
\bottomrule
\end{tabular}
}
\end{table}
\vspace*{-0.6cm}

\begin{figure}[h!]
\centering
\includegraphics[width=0.55\columnwidth]{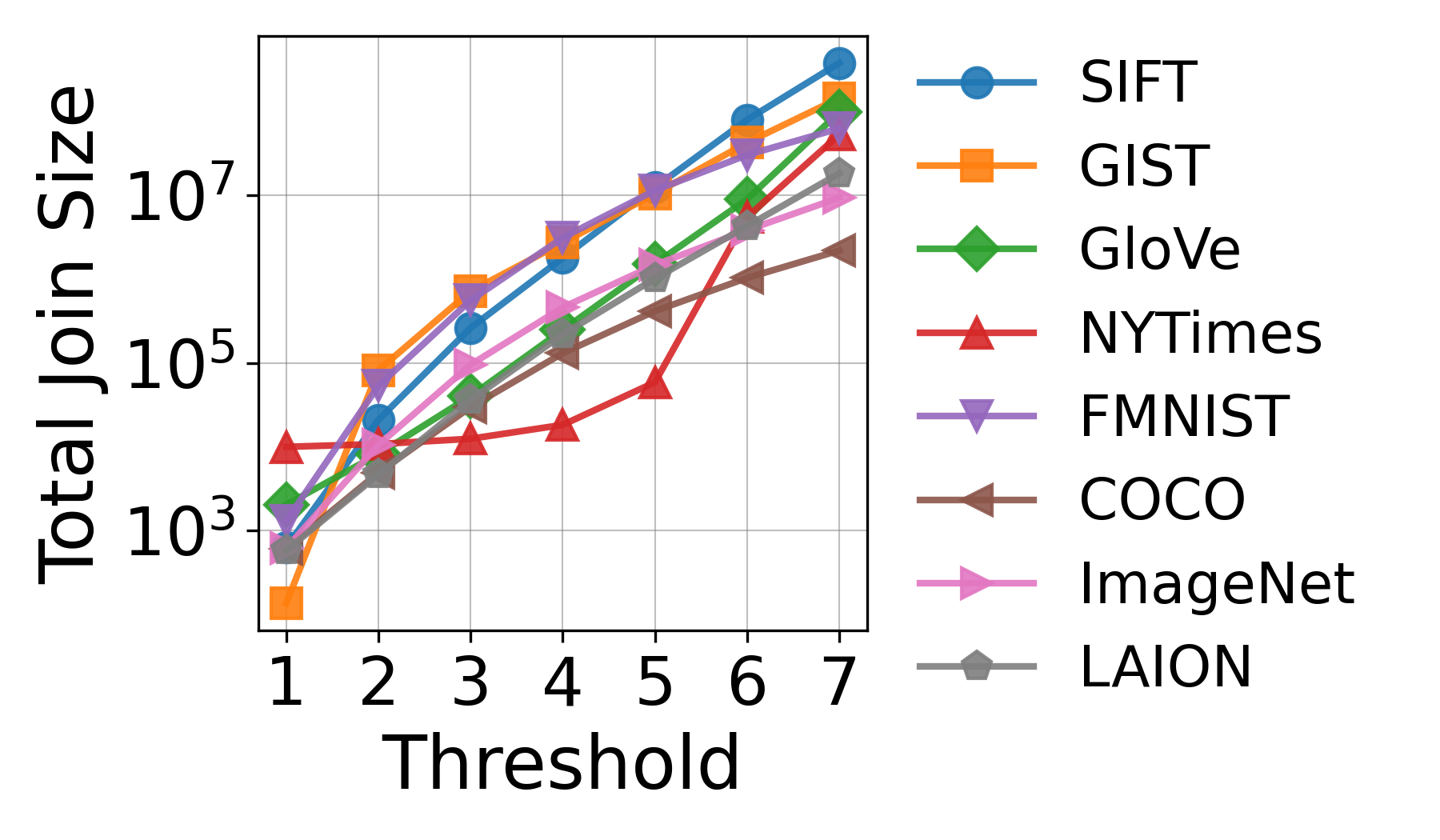}
\vspace*{-0.4cm}
\caption{Join size per dataset and threshold.}\label{fig:dataset_join_size}
\vspace*{-0.4cm}
\end{figure}


\subsubsection{Baselines}



We implemented all baselines within our unified vector join framework to ensure a fair comparison and isolate the effect of \revisionA{}{each component.} 
We also implemented the state-of-the-art threshold-based vector join method, \SimJoin\footnote{\revisionA{}{We chose \SimJoin as the main competitor that showed superior performance to \XJoin (Section \ref{sec:related}), and neither XJoin nor DiskJoin is open-sourced.}} \cite{SimJoin}, inside our framework and achieved similar efficiency with their paper. 
This allows us to fairly assess how our techniques improve over state-of-the-art work sharing.


\begin{itemize}[topsep=0pt, leftmargin=0.2cm]




\revisionA{}{

\item \NLJ: Nested-loop join that computes all pairwise distances between queries and data, which is exact but impractical.

\item \INLJ: Index nested-loop join, each query independently performs an approximate search over the data index (\texttt{cfg.merged = false, work\_sharing = none} in Section \ref{subsec:ours:framework}).


\item \HWS: Hard work sharing (\texttt{work\_sharing = hard}) on separate indexing, corresponding to our implementation of \SimJoin \cite{SimJoin}.

\item \SWS: Our soft work sharing (\texttt{work\_sharing = soft}) on separate indexing.

\item \MI: Our merged index (\texttt{merged = true, work\_sharing = self}).

\item \MIA: Merged index with our OOD-aware adaptive hybrid search (\texttt{adapt\_hybrid = true}).

}

\end{itemize}

We use the NSG graph index \cite{NSG} following \cite {SimJoin} due to its strong efficiency-recall trade-off and fast convergence. Index parameters (e.g., graph degree and max queue size) are kept consistent across all methods to ensure comparability. We use the default parameters provided in \cite{NSGCode} \kkim{with max neighborhood size of 70 and max queue size of 256.}
We do not extensively tune the underlying index search algorithms beyond standard parameter settings, such as applying vector quantizations or aggressive SIMD optimizations, as our goal is to evaluate join-level optimizations rather than single-query ANN performance.

\vspace*{-0.2cm}
\subsubsection{Metrics}


We measure the latency of join execution, recall, and memory usage for the online join processing. We measure indexing time and index size for the offline phase.






\vspace*{-0.2cm}
\subsubsection{System Configuration}

For all our experiments, we use a server with Intel(R) Xeon(R) Gold 5118 CPUs and 376 GB memory. We use a single thread for evaluating vector join.

\vspace*{-0.2cm}
\subsection{Vector Join Performance}\label{subsec:exp:base}


\subsubsection{Overall Results.} \label{overall_results}



Figure \ref{fig:main_experiment} shows online vector join performance. 
Overall, our proposed techniques -- Soft Work Sharing (\SWS), Merged Index (\MI), and Adaptive Hybrid Search (\MIA) -- consistently outperform the state-of-the-art baseline Hard Work Sharing (\HWS) in both efficiency and robustness of recall. 

\begin{figure*}[!t]
    \centering
    \includegraphics[width=0.9\textwidth]{exp_figure/combined_results_median_square_tight.png}
    \vspace{-10pt}
    \caption{Vector join performance (latency, recall, and memory usage) on seven thresholds and eight datasets.}\label{fig:main_experiment}
    \vspace*{-0.2cm}
\end{figure*}






\noindent \textbf{\underline{\NLJ.}} As expected, \NLJ shows consistently high latency across all thresholds and gives perfect recall.

\noindent \textbf{\underline{\INLJ.}} 
As expected, \INLJ is significantly faster than \NLJ at small thresholds, but may drop the recall noticeably for certain datasets. For \NYTimes and \GIST, the greedy search cannot find an in-range point due to weak data locality. The index is weakly connected (node degrees are much smaller than the other datasets, i.e., smaller Mode in Table \ref{tab:datasets}) and does not effectively navigate to existing in-range points.
This also occurs for OOD datasets. 






\noindent \textbf{\underline{\HWS and \SWS.}} 
\sigmod{While \HWS reduces latency by up to 1.6x over \INLJ,} \SWS achieves lower latency than \HWS across all datasets, especially for small thresholds ($\theta_1$-$\theta_3$), where \HWS often caches no point at all and thus fails to reuse traversal efforts. By caching the closest out-range point, \SWS significantly reduces the greedy search overhead for similar queries. This improvement is most visible on \SIFT, \GIST, \Glove, and \FMNIST, where \SWS consistently reduces execution time by up to \sigmod{3x} compared to \HWS, while maintaining similar recall.
For large thresholds, \SWS significantly reduces the memory usage of \HWS by caching the closest data node per query, not all in-range nodes redundantly.


\noindent \textbf{\underline{\MI.}} \MI significantly reduces greedy search cost by offloading the task of finding an initial in-range point to the offline index construction. As shown in Figure \ref{fig:main_experiment}, \MI achieves the lowest latency among all methods across most thresholds and datasets. This effect is particularly strong for small thresholds, where the greedy search in separate indexes often traverses many nodes before finding the first in-range point. With \MI, the search frequently starts directly from a close data neighbor, reducing traversal overheads. Consequently, MI approaches constant-time behavior for finding an in-range point, yielding substantial speedups over \HWS and \SWS \sigmod{by up to 44.1x and 26.9x} without sacrificing recall. \kkim{Furthermore, enabling direct search from close data points often increases recall by a large margin, such as for \GIST and \NYTimes. 
Still, it often shows low recall for OOD queries; while being higher than separate-index methods, but not yet sufficient without hybrid search.} 


\noindent \textbf{\underline{\MIA.}}
For OOD datasets (\COCO, \ImageNet, and \LAION), \MIA substantially improves recall by up to 43\%. While BFS-expansion can miss disconnected in-range regions for OOD queries, \MIA selectively activates the hybrid search strategy to bridge out-range walls. As a result, it consistently achieves higher recall than the others on these datasets with a trade-off of higher latency than \MI. Still, the latency is lower than using separate indexes (\SWS and \HWS). For non-OOD datasets, \MIA achieves similar latency with \MI by adaptively disabling the hybrid search and falling back to BFS (Table \ref{tab:datasets}). 
\nata{Together, these results validate the robustness claim from Section \ref{sec:intro}: the co-design maintains high recall across all eight datasets, including the structurally challenging OOD distributions.}

\revisionA{R3.W4}{

\noindent \textbf{\underline{Other methods.}}
We omitted other methods, e.g., \SWSH  (\texttt{force\_hybrid = true}), \MIH, and \MISWS \sigmod{whose behaviors are predictable from our analysis and showed smaller performance gaps compared to the others.} For \SWSH and \MIH, forcing hybrid search for all queries increases latency for all datasets while improving recall for only OOD datasets due to the lack of adaptivity.
For \MISWS, reusing the closest visited point per query does not lead to substantial latency reduction as in \INLJ vs. \SWS since \MI already significantly reduces greedy search steps.
Still, we evaluate \MISWS in Section \ref{subsec:exp:equi_size_join} and show that it can improve recall when \MI misses in-range regions.

}



\subsubsection{Latency-recall Trade-off.}

Figure \ref{fig:latency_recall_exp} presents the latency–recall trade-off when varying the max queue size \revisionA{}{(\texttt{greedy\_beam\_size} or \texttt{hybrid\_beam\_size} in Section \ref{subsec:ours:framework})} for the smallest threshold on all eight datasets. 
\revisionA{}{By default}, the queue size controls the greedy search, while for \MIA, it affects the hybrid search. 
\revisionA{}{As mentioned in Section \ref{subsec:ours:framework}, the greedy search is similar to top-1 search and thus is not affected by queue size much.}

While \MI and \MIA together form the best Pareto curves, \MIA shows a clear latency-recall trade-off by controlling the hybrid queue size. This improves recall for OOD-heavy datasets by allowing controlled traversal through out-range nodes.

\begin{figure*}[!t]
\centering
\includegraphics[width=\textwidth]{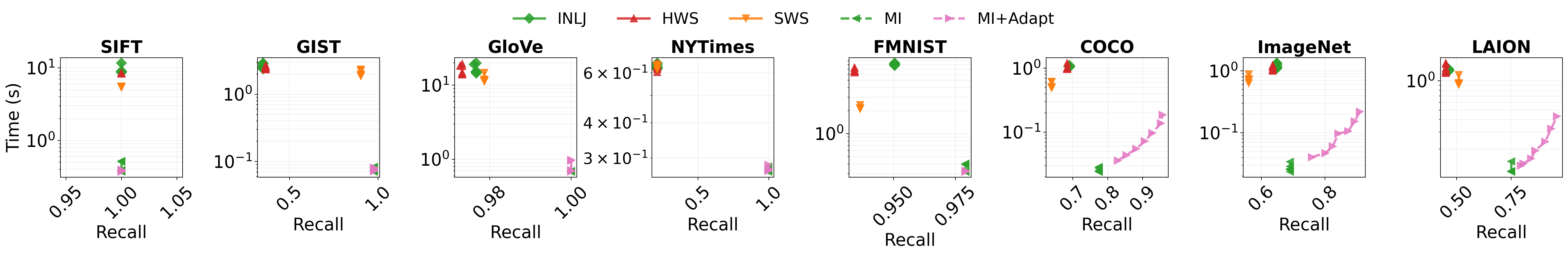}
\vspace*{-0.8cm}
\caption{Latency-recall curves for varying queue sizes from 8 to 512. Smallest threshold, $\theta_1$, is used. The queue size is defined for the greedy search phase for non-MI methods and hybrid search phase for \ESMIA. 
}\label{fig:latency_recall_exp}
\vspace*{-0.2cm}
\end{figure*}

\subsubsection{Latency Breakdown.}

Figure \ref{fig:profiling_exp} breaks down the total join latency into \kkim{three} parts: greedy search, BFS (or hybrid search), and the others. 
As expected, \INLJ, \HWS, and \SWS are bottlenecked by the greedy search for small thresholds, finding a single in-range point, and the main goal of \revisionA{}{work sharing and merged indexing} has been to reduce the traversal effort in the greedy search. \revisionA{}{\MI clearly resolves this overhead.}
BFS dominates runtime for large thresholds due to large join sizes, seeking for all in-range points. 

\begin{figure}[h!]
\vspace*{-0.3cm}
    \centering
    \includegraphics[width=0.85\columnwidth]{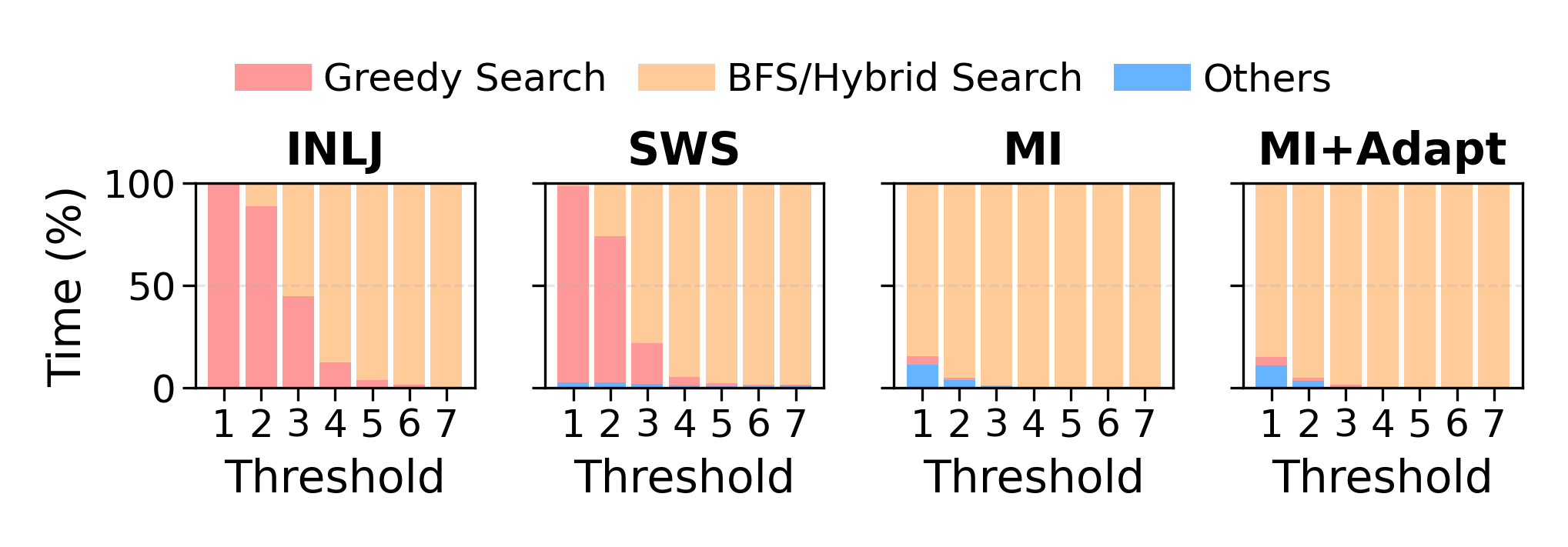}
    \vspace*{-0.4cm}
    \caption{Latency breakdown for \FMNIST into three: greedy search, BFS (or hybrid search), and the others. Similar results observed for the other datasets. HWS is similar to SWS.}\label{fig:profiling_exp}
    \vspace*{-0.4cm}
\end{figure}

\subsection{Query-Level Analysis} \label{subsec:exp:query_level_analysis}

\revisionA{R3.W2}{This section provides a more in-depth query-level analysis. Since per-query latency is often less than 1ms and hard to measure stably, we replace the latency with the number of distance computations (or simply steps) which is the main bottleneck.}

\subsubsection{Hard vs. Soft Work Sharing}

\revisionA{}{
Figure~\ref{fig:sws_vs_hws_matrix} groups queries by relative recall and steps for \SWS vs. \HWS. On \SIFT, both methods achieve identical recall for all queries, while \SWS uses fewer steps for most of them. On \FMNIST, the same pattern largely remains: most queries still have identical recall, and \SWS usually requires fewer steps, although a small number of queries favor either \SWS or \HWS in recall.

Figure~\ref{fig:sws_vs_hws_min_dist_so_far} explains this behavior using the evolution of the minimum computed distance over steps on \FMNIST. In both query groups, \SWS generally starts closer to the query than \HWS, indicating that soft sharing provides a better local initialization by reusing out-range points as well.
However, for queries where \SWS{} has lower recall, this local advantage does not always translate into better final coverage. In such cases, \HWS{} can occasionally reach additional in-range points through larger jumps in the graph, but these jumps are artifacts of the underlying index structure rather than the objective of hard work sharing itself. This suggests that the remaining gap is not fundamentally about stronger work sharing, but about access to a better \emph{global} index structure, which is orthogonal to work sharing.


\begin{figure}[h!]
\vspace*{-0.3cm}
    \centering
    \includegraphics[width=0.98\columnwidth]{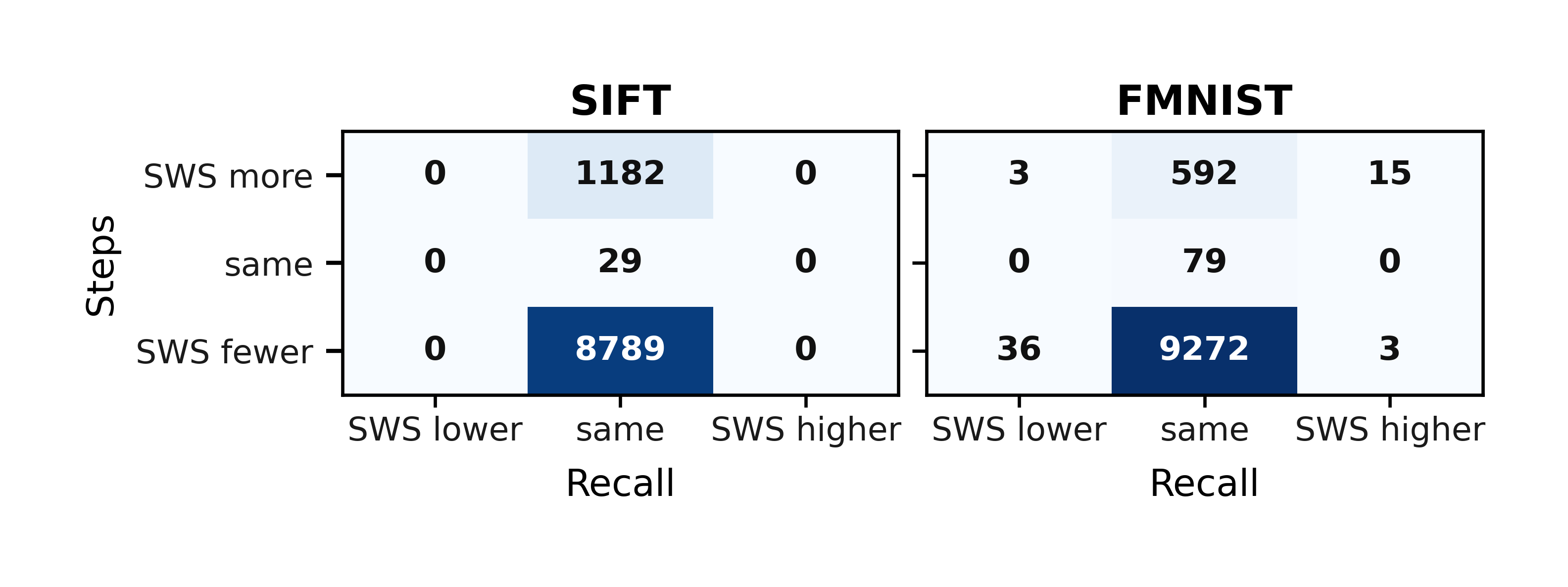}
    \vspace*{-0.8cm}
    \caption{Query frequencies for \SWS vs. \HWS on \SIFT and \FMNIST, using the smallest threshold $\theta_1$.}\label{fig:sws_vs_hws_matrix}
    \vspace*{-0.4cm}
\end{figure}

\begin{figure}[h!]
\vspace*{-0.4cm}
    \centering
    \includegraphics[width=0.98\columnwidth]{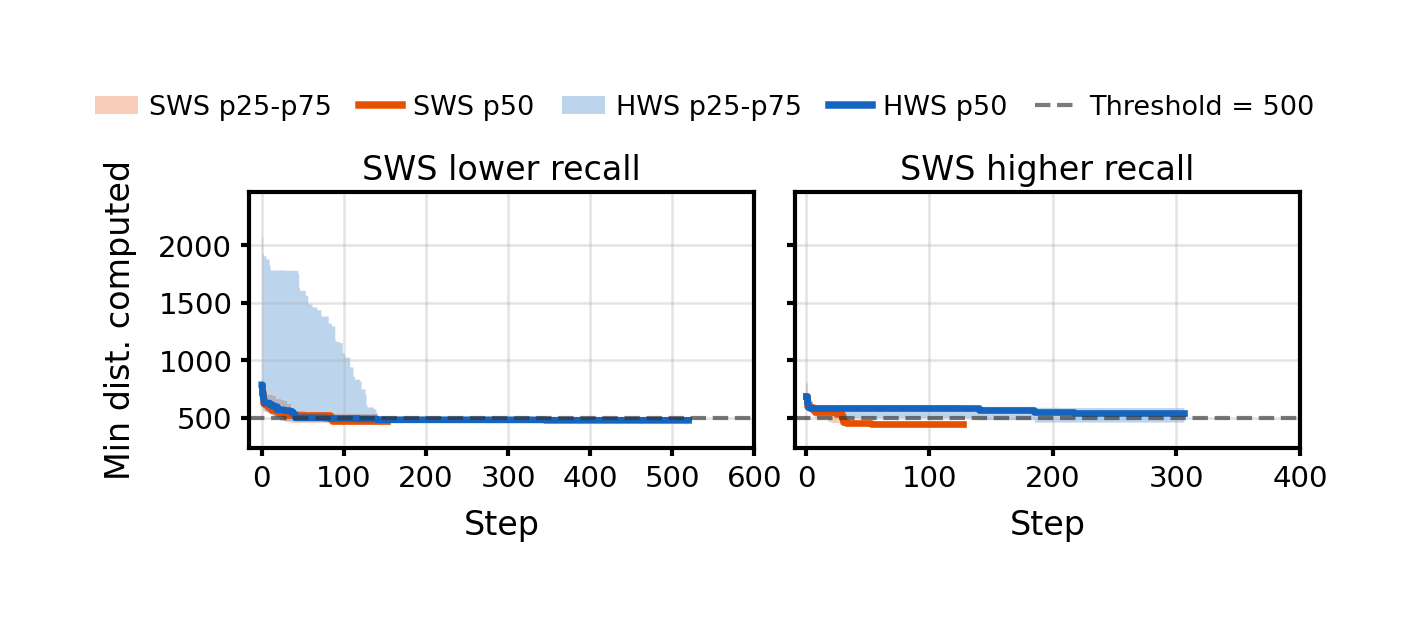}
    \vspace*{-0.8cm}
    \caption{Min. distance computed over steps during query execution, aggregated across queries where \SWS has lower/higher recall than \HWS, on \FMNIST in Figure \ref{fig:sws_vs_hws_matrix}.}\label{fig:sws_vs_hws_min_dist_so_far}
    \vspace*{-0.4cm}
\end{figure}

}

\subsubsection{Separate vs. Merged Index}

\revisionA{}{

We next compare the separate-index design (\INLJ{}) and the merged index (\MI{}) at the query level. Figure~\ref{fig:mi_vs_si_matrix} shows that on \SIFT, both methods achieve identical recall for all queries, while \MI{} uses fewer steps for every query. On \FMNIST, the same pattern largely remains: \MI{} preserves recall for most queries and still requires fewer steps in all cases. This confirms that the primary benefit of \MI{} is to reduce online search effort by starting from a nearby data-side region already exposed by the merged graph, rather than repeatedly discovering it online.

For most \FMNIST queries where \MI{} has lower recall than \INLJ{}, \MI{} still finds at least one match very quickly, but misses few additional matches because the in-range regions are disconnected. Thus, \MI{} does not fail to reach an in-range region; rather, different seeds lead \MI{} and \INLJ{} to different in-range components. Conversely, when \MI{} has higher recall, it again benefits from the same local seed advantage.}
\revisionA{R3.D3}{These results are consistent with our observations on \ImageNet. Even when query and data distributions barely overlap, only 1\% of queries fail to have any data point in their local merged-index neighborhoods. \MI{} almost always reduces search steps, while the recall losses arise from disconnected in-range regions. This motivates the adaptive hybrid traversal for structurally difficult queries while trading off with more steps.}

\begin{figure}[h!]
\vspace*{-0.35cm}
    \centering
    \includegraphics[width=0.98\columnwidth]{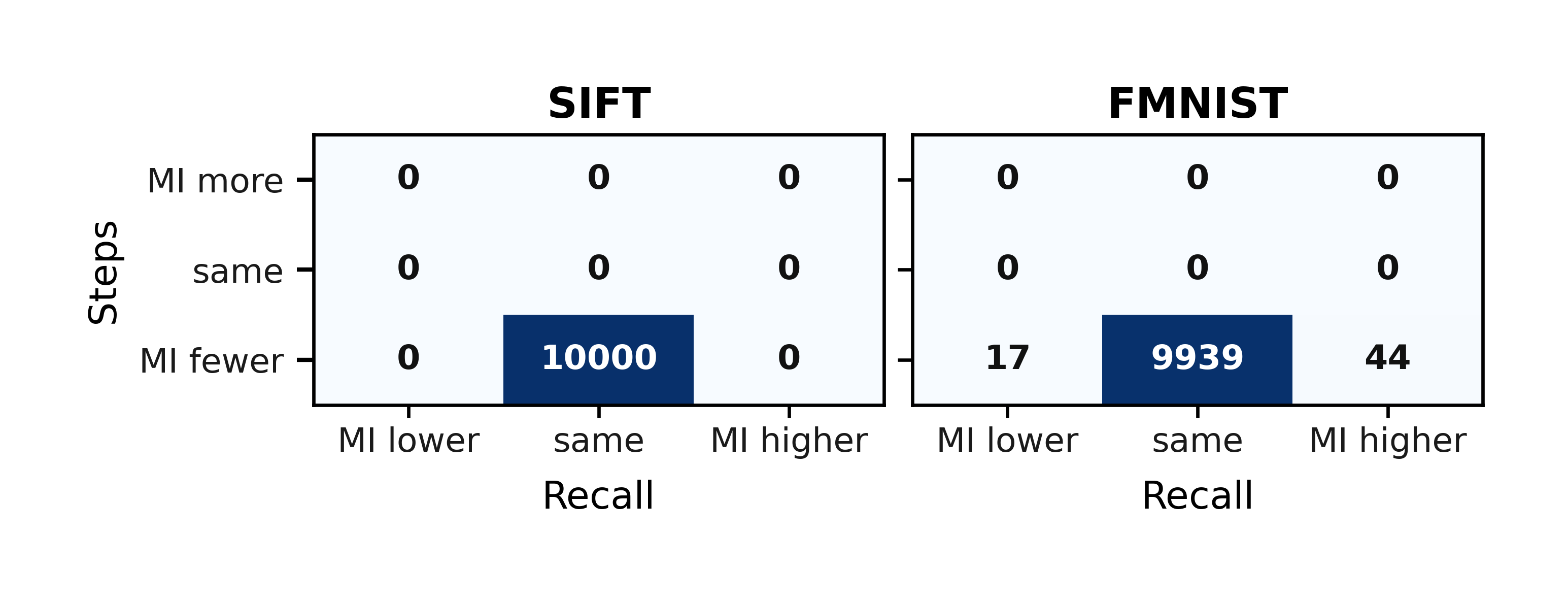}
    \vspace*{-0.8cm}
    \caption{Query frequencies for \MI vs. \ES.}\label{fig:mi_vs_si_matrix}
    \vspace*{-0.6cm}
\end{figure}

\subsection{Equi-Size Join} \label{subsec:exp:equi_size_join}


\revisionA{R2.D3}{

We sample 1M query and 1M data vectors from SIFT1B\footnote{http://corpus-texmex.irisa.fr} dataset to evaluate the methods under $|X|$ as large as $|Y|$,
and four smallest thresholds from \SIFT since all methods have high latencies for larger thresholds.
We call this \SIFTE. In Figure \ref{fig:SIFT_E}, \SWS again reduces latency by up to 2.5x compared to \HWS without sacrificing recall. 

Notably, \MI{} again significantly reduces latency by 44x but shows lower recall of 0.87 than \INLJ's 0.95 at the smallest threshold (for larger thresholds, both achieve high recall close to 1.0) due to the large query set that acts similarly to the out-range wall in the BFS expansion for OOD datasets.
Adding \SWS{} on top of \MI{} (i.e., \MISWS) improves recall to 0.9, which differs from the usual effect of \SWS{} on separate indexes, where it mainly reduces steps. We attribute this to the strong overlap between query and data distributions in \SIFTE{}: the nearest data point reused across queries is often already an in-range seed for another query, so \SWS{} helps \MI{} recover additional relevant regions at almost no extra cost. In fact, it also yields a modest latency reduction of up to 1.2x, likely due to CPU cache effects from processing similar queries in sequence in the MST-based ordering. On the other datasets, where query-data overlap is weaker, this effect is much less pronounced, and \SWS{} on top of \MI{} brings smaller benefit. This suggests that \MI{} and \SWS{} are complementary: \MI{} provides a better structural organization of the search space, while \SWS{} exploits online reuse within it.

\begin{figure}[h!]
\vspace*{-0.2cm}
    \centering
    \includegraphics[width=0.72\columnwidth]{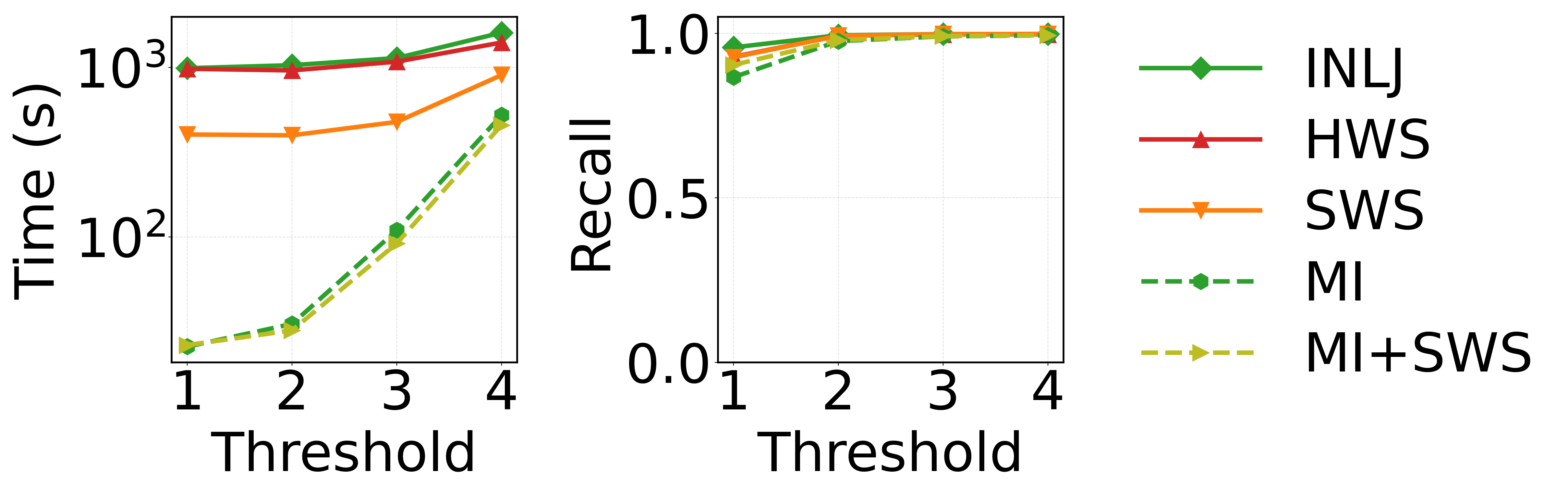}
    \vspace*{-0.4cm}
    \caption{Vector join performance on \SIFTE using four thresholds $\theta_1$-$\theta_4$.}\label{fig:SIFT_E}
    \vspace*{-0.2cm}
\end{figure}

}

\vspace*{-0.2cm}
\subsection{Offline Overhead} \label{subsec:exp:offline_overhead}

Figure \ref{fig:index_building} compares the index size and build time of using two separate indexes -- one for queries $G_X$ and one for data $G_Y$ -- against merged index $G_{X \cup Y}$. Across all datasets, merged index does not introduce much overheads compared to separate indexes, \sigmod{up to additional 4.5\% size and 0.7\% time}, because the merged index uses the same graph structure and hyperparameters (e.g., max neighborhood size), and the total number of vectors is the same as $|X|+|Y|$ (\revisionA{R1.D5}{refer to Appendix B for theoretical analysis}).
Since index construction is an offline process performed once per dataset, small additional overhead is negligible relative to the substantial runtime improvements achieved during online vector joins.

\begin{figure}[h!]
\vspace*{-0.2cm}
    \centering
    \includegraphics[width=0.96\columnwidth]{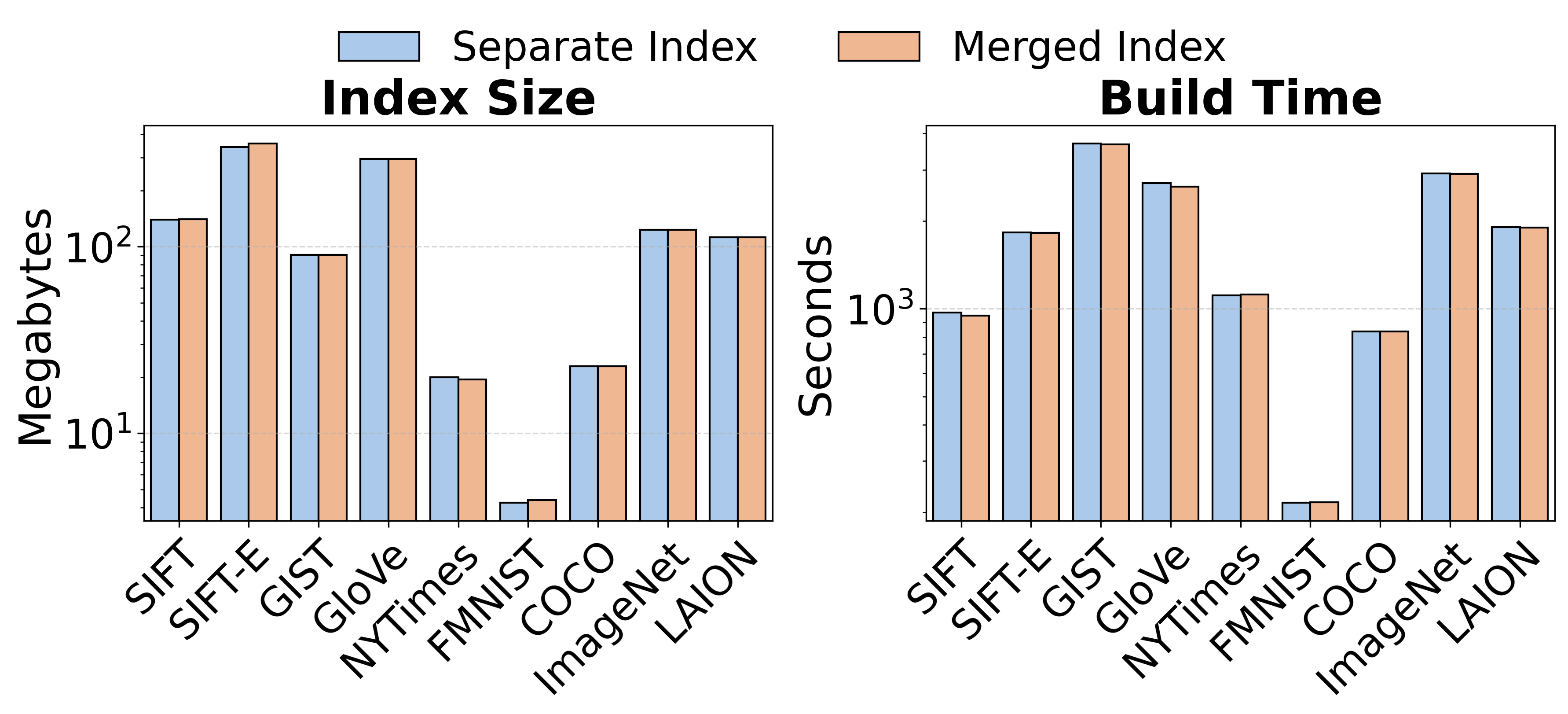}
    \vspace*{-0.3cm}
    \caption{Offline overheads (index size and indexing time) for separate query and data indexes vs. merged index.}\label{fig:index_building}
\end{figure}











\vspace*{-0.2cm}
\subsection{Scalability Test}\label{subsec:exp:scalability}


To assess methods under varying data scales, we vary the number of data vectors $|Y|$ using the SIFT1B dataset and the smallest threshold, while keeping the number of queries at 10K. We vary $|Y|$ from 10K to 10M and sample $|Y|$ vectors from SIFT1B. 
10M is the max data size we can process in memory. \kkim{We leave supporting billion-scale data using multiple servers or disks as a future work.}
Figure \ref{fig:scalability_exp} reports join latency as $|Y|$ increases. Recall is near 1.0 for all methods.

\NLJ scales linearly with $|Y|$, as expected, since it computes all pairwise distances. 
\revisionA{}{\INLJ scales better than \NLJ \nata{(each query takes $O(\log |Y|)$ \cite{HNSW, NSG})}}
but still exhibits noticeable growth due to repeated greedy searches. \HWS reduces some redundancy, but its performance degrades for large $|Y|$ due to increasing cache sizes and BFS expansion costs.
\SWS shows improved scalability by further reducing redundant greedy search across similar queries by caching closest out-range points. 
\nata{MI achieves the best scalability: as predicted by the constant-time in-range region discovery of Section \ref{subsec:ours:mi}, its greedy search latency scales near-horizontally, confirming the practical transition from the $O(\log |Y|)$ traversal to an $O(1)$ neighborhood probe as the data size grows.}


\begin{figure}[h!]
\vspace*{-0.2cm}
\centering
\includegraphics[width=0.65\columnwidth]{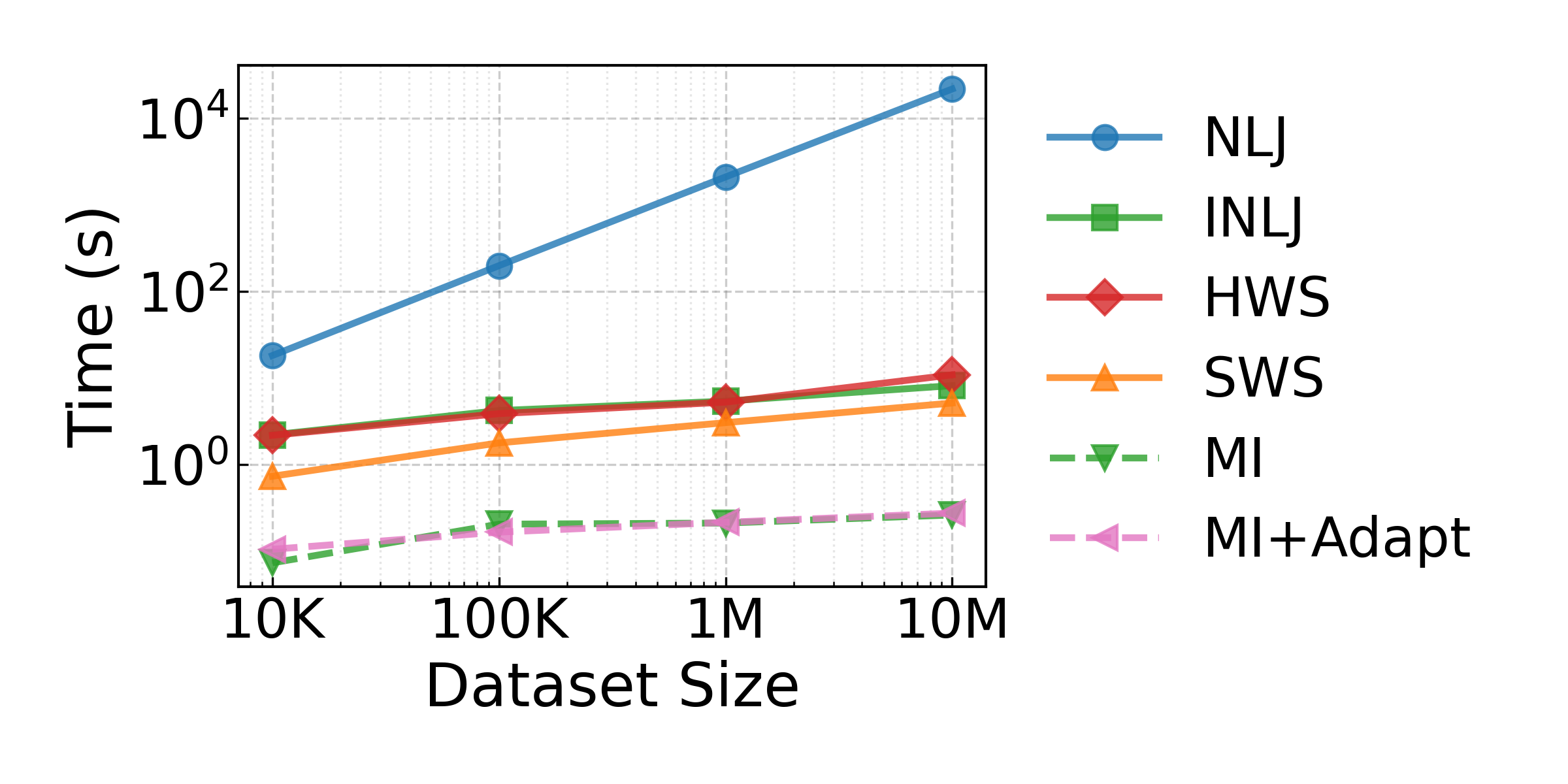}
\vspace*{-0.5cm}
\caption{Scalability test with varying number of data vectors, using the smallest threshold $\theta_1$.
}\label{fig:scalability_exp}
\vspace*{-0.2cm}
\end{figure}




\vspace*{-0.2cm}
\subsection{Varying Index Type}\label{subsec:exp:ablation}



So far, our experiments have used the NSG graph index \cite{NSG} following \cite{SimJoin}. In this section, we evaluate whether our techniques generalize to a different index structure, HNSW \cite{HNSW}, which is also widely used in modern vector search systems \cite{douze2024faiss, wang2021milvus}.

Figure \ref{fig:index_comparison} compares the performance of all methods using HNSW graphs instead of NSG on \FMNIST and \kkim{\ImageNet}. 
The ranking of methods remains consistent:
\MI and \MIA achieve the best latency-recall trade-offs, followed by \SWS, \HWS, and \INLJ. This shows that our \revisionA{}{framework and techniques} are index-agnostic and not tailored specifically to NSG. \todo{While we set the HNSW index sizes similar to the NSG indexes we used, in general HNSW gives lower performance (especially recall) than NSG. This supports the choice of using NSG by default in \cite{SimJoin}.}





\begin{figure}[h!]
\vspace*{-0.2cm}
\centering
\includegraphics[width=0.92\columnwidth]{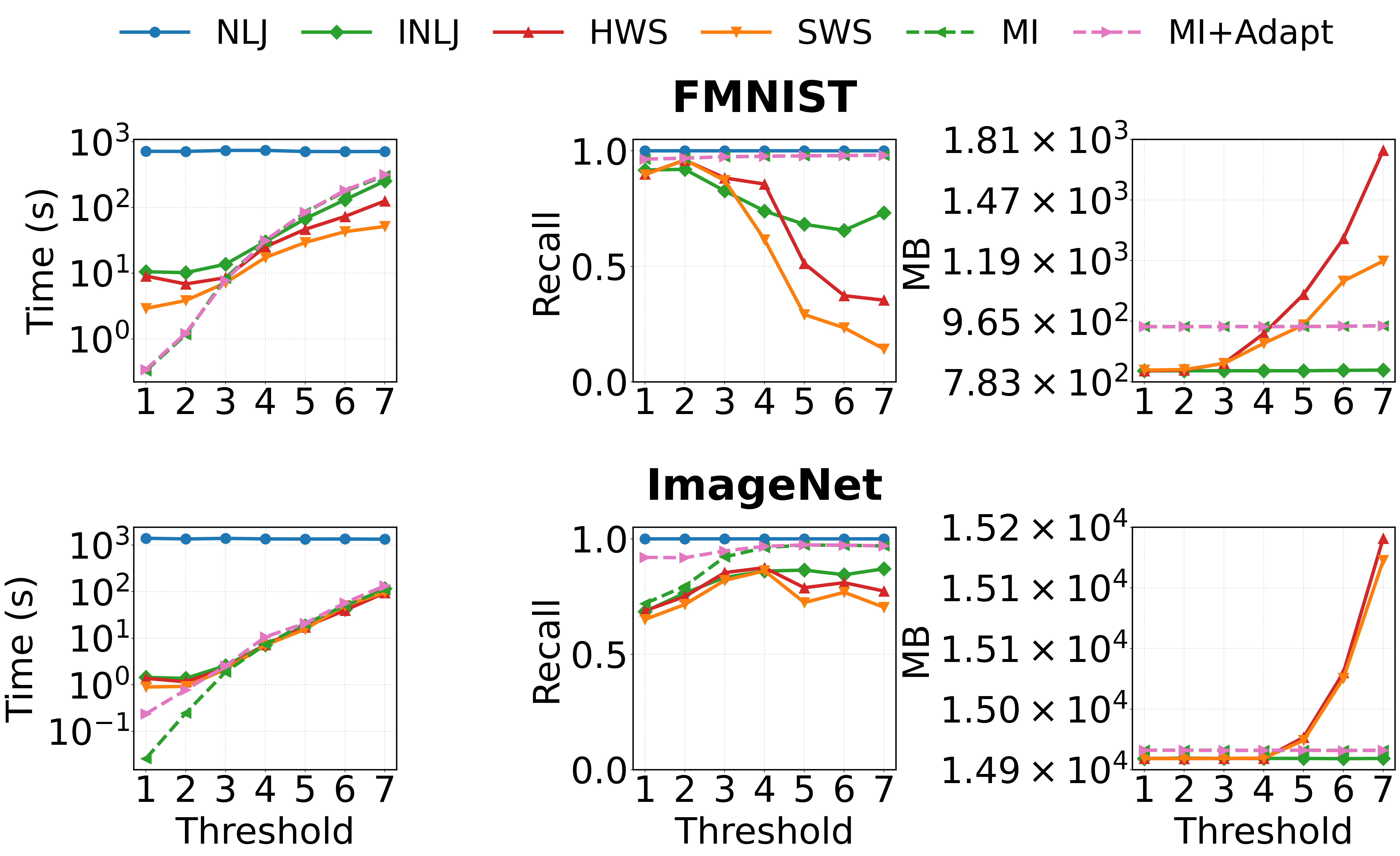}
\vspace*{-0.3cm}
\caption{Vector join performance using HNSW index \cite{HNSW} on \FMNIST and \ImageNet.}\label{fig:index_comparison}
 \vspace*{-0.2cm}
\end{figure}


\begin{figure*}[t]
\centering
\begin{minipage}[b]{0.22\textwidth}\centering
\includegraphics[width=\linewidth]{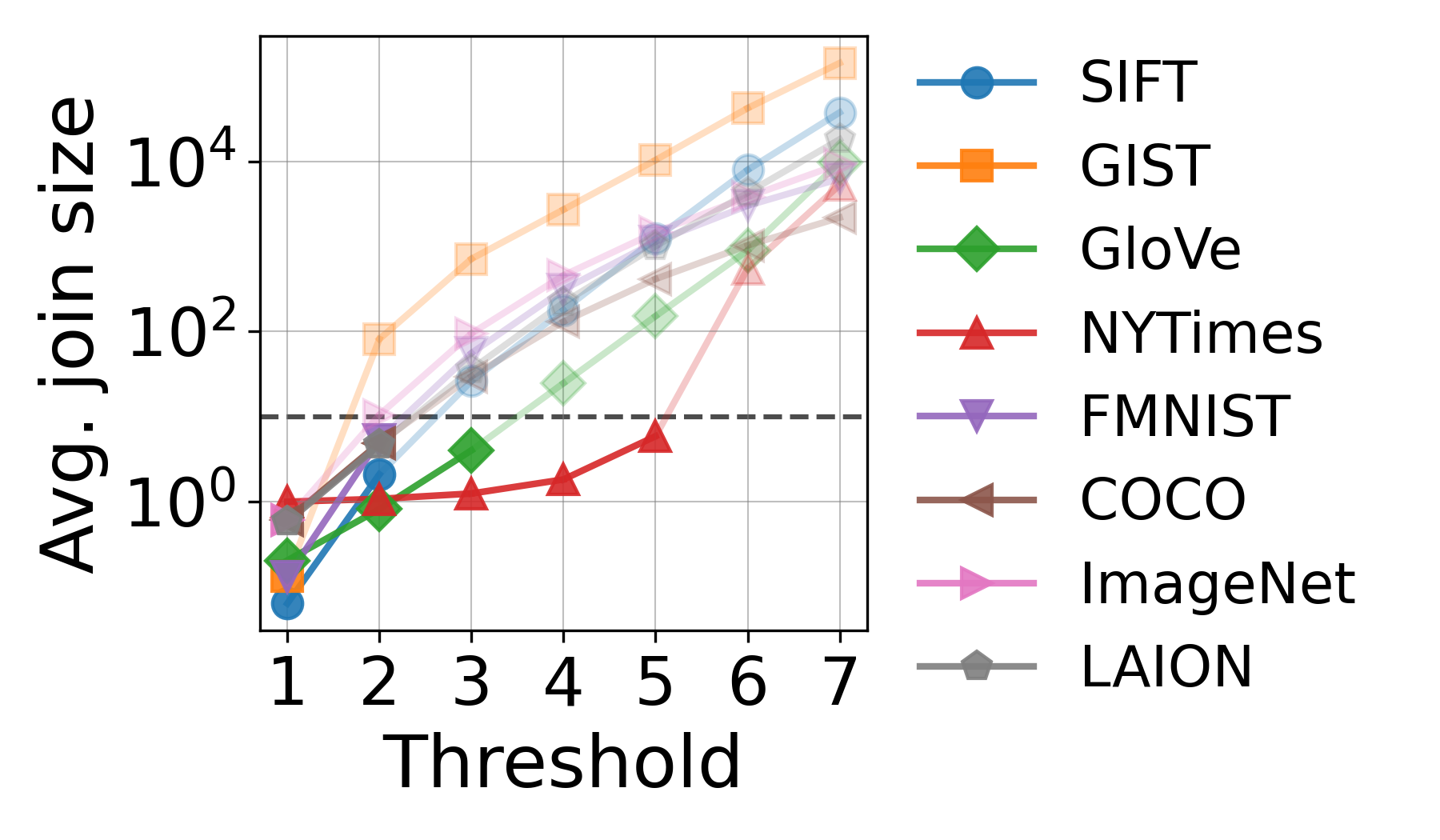}\\[-7pt]
{\footnotesize (a)}
\end{minipage}\hfill
\begin{minipage}[b]{0.265\textwidth}\centering
\includegraphics[width=\linewidth]{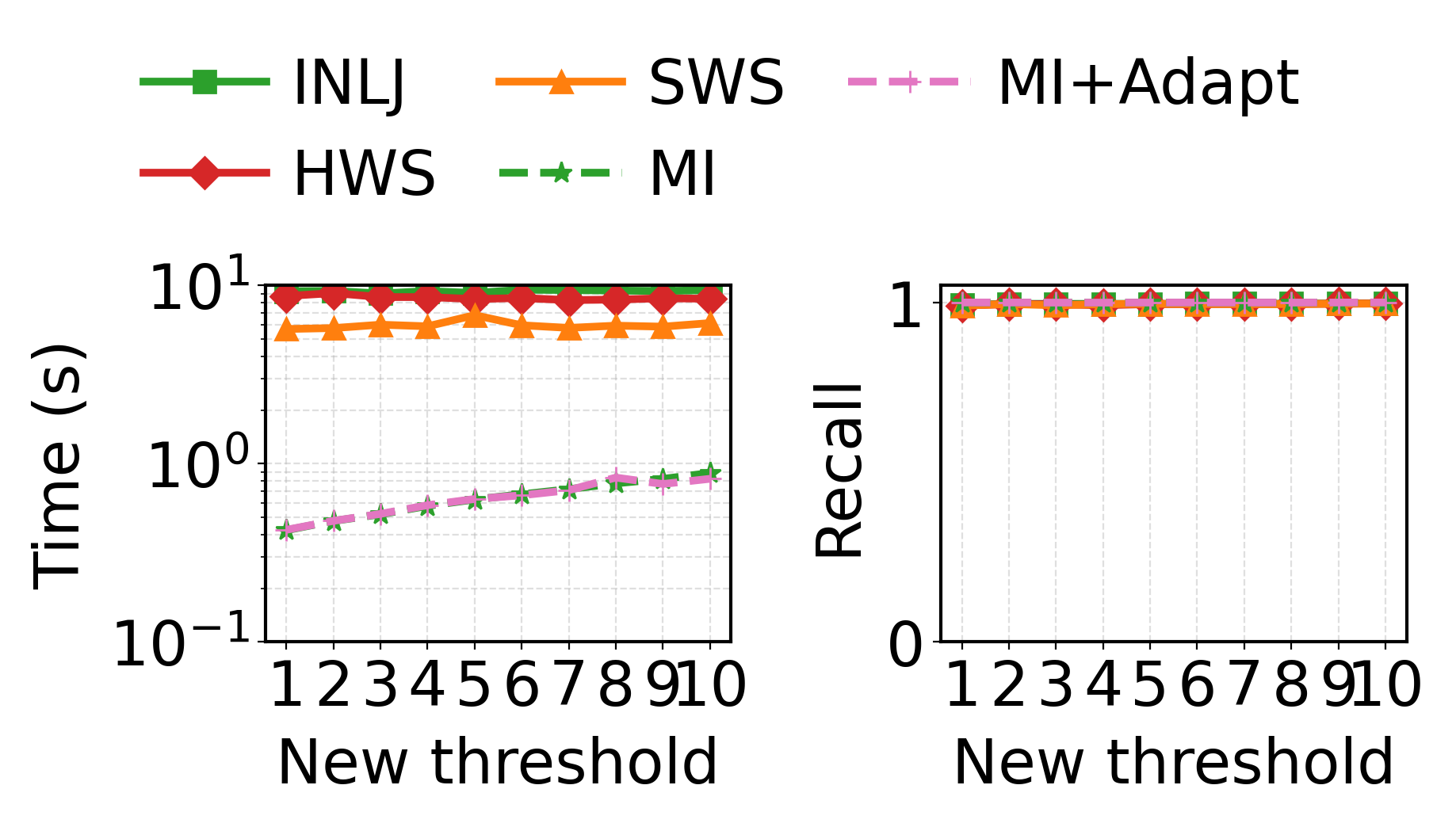}\\[-7pt]
{\footnotesize (b)}
\end{minipage}\hfill
\begin{minipage}[b]{0.24\textwidth}\centering
\includegraphics[width=\linewidth]{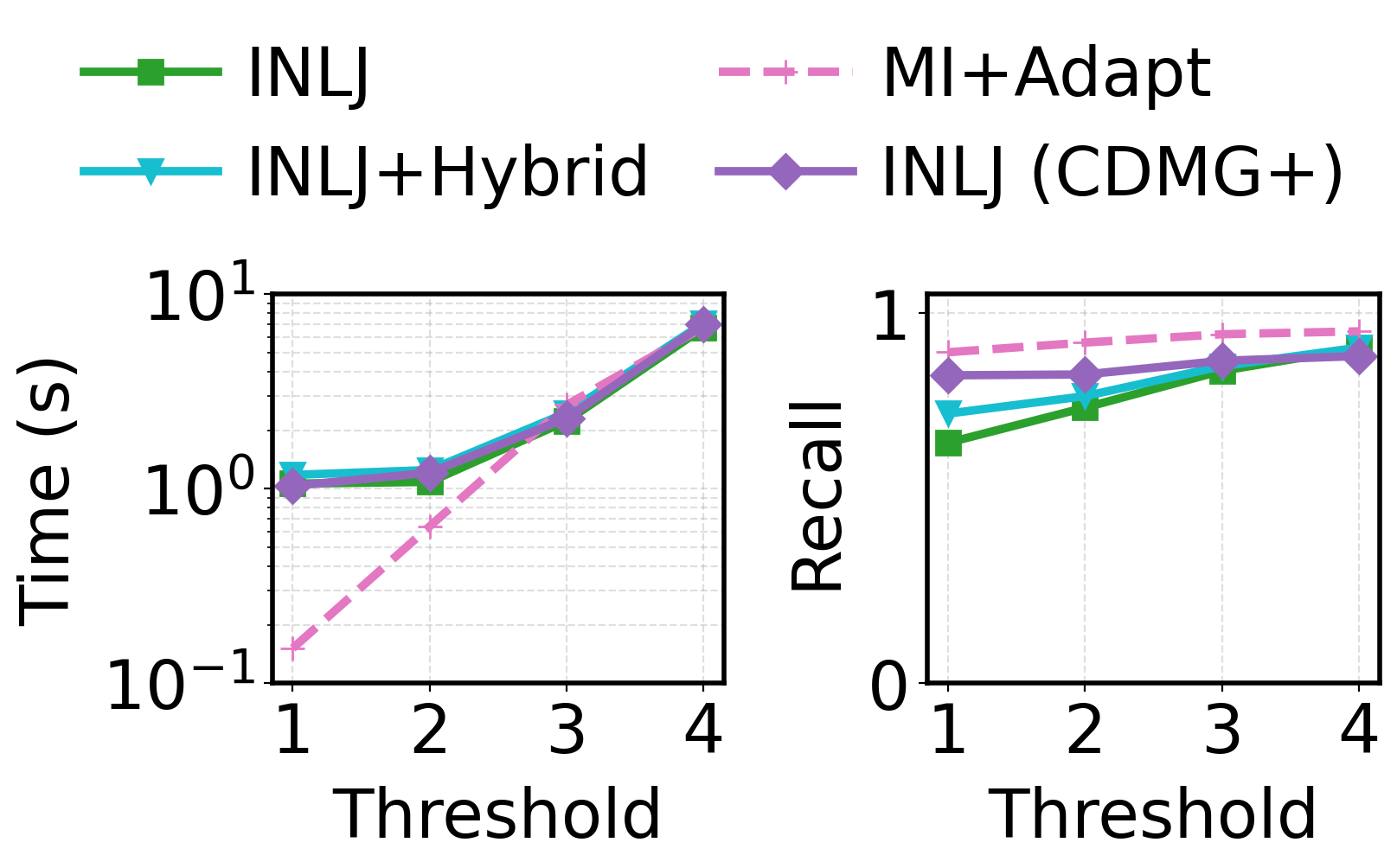}\\[-7pt]
{\footnotesize (c)}
\end{minipage}\hfill
\begin{minipage}[b]{0.253\textwidth}\centering
\includegraphics[width=\linewidth]{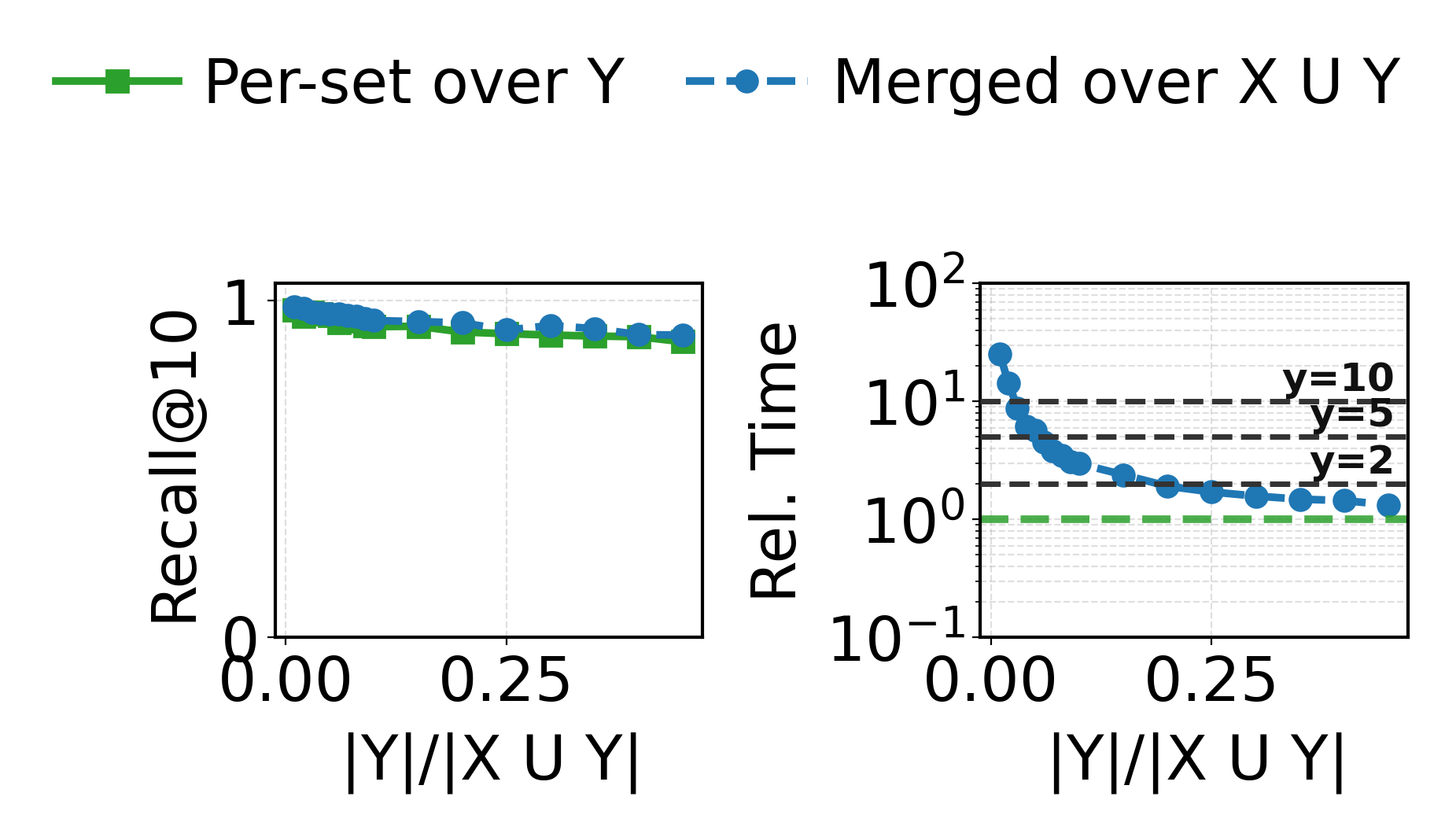}\\[-7pt]
{\footnotesize (d)}
\end{minipage}
\vspace*{-0.35cm}
\caption{\moved{
\textbf{(a)} Avg.\ join size $k$ per query vs.\ our thresholds in Table \ref{tab:epsilon_mapping}; the band $k \leq 10$
is \SimJoin's~\cite{SimJoin} regime.
\textbf{(b)} Finer \SIFT\ thresholds calibrated to $k=1,\dots,10$:
\MI/\MIA\ stay about an order of magnitude faster than \INLJ/\HWS/\SWS\ at recall
$\approx$1.
\textbf{(c)} On the OOD dataset \ImageNet, \INLJ\ with forced hybrid search (no adaptivity based on the merged-index signal, Section \ref{subsec:ours:ood_hybrid}) and
\INLJ\ over the state-of-the-art OOD graph index CDMG+ \cite{CDMG} both improve recall but
remain below \MIA.
\textbf{(d)} Single-set top-$10$ search (ANNS) on a per-set index over $Y$ vs.\ the merged index over $X\cup Y$ on
\SIFT\ (merged index fixed at $1$M vectors): recall matches while the merged-index latency overhead decreases from $\sim$25$\times$ to $\sim$1.3$\times$ as $|Y|/|X\cup Y|$ grows from 1\% to 45\%.}}
\label{fig:revision}
\vspace*{-0.3cm}
\end{figure*}

\subsection{Threshold Regimes and Join Size}\label{subsec:exp:revision_thresholds}

\moved{\textbf{Answer to: Why all do methods show similar performance?} Small and large thresholds are two genuinely different use cases, and the dominant one is \emph{small} thresholds with small join sizes---the analogue of small $k$ in top-$k$ search \cite{HNSW, NSG}; \SimJoin~\cite{SimJoin} likewise tunes its thresholds to $\leq$10 results per query. Figure \ref{fig:revision}(a) puts our threshold design in this perspective: our smallest thresholds per dataset correspond to \SimJoin's regime, while the larger ones deliberately extend the sweep for completeness. In the small-threshold regime, work sharing and merged indexing target the greedy-search overhead, which is the main bottleneck (Figure \ref{fig:profiling_exp}), and our methods win decisively. To test this regime directly, we re-ran \SIFT\ with ten finer thresholds calibrated to an average join size of exactly $k=1,2,\dots,10$: \MI\ and \MIA\ stay about an order of magnitude faster than \INLJ, \HWS, and \SWS\ at recall $\approx$1 throughout (Figure \ref{fig:revision}(b)).}

\moved{\emph{Large} thresholds are a distinct, under-explored regime. There, the post-greedy expansion dominates \emph{every} method, which must enumerate all (up to $>10^4$) results per query (Figures \ref{fig:revision}(a) and \ref{fig:profiling_exp}), so the methods converge; this degradation is universal---\cite{DiskJoin} also reports diminishing gains at large join sizes, and very recent work targets large-$k$ search specifically \cite{yin2026bbc}---and we do not claim to solve it. We report this unfavorable regime for a fair, complete comparison rather than showing only favorable cases (Appendix E discusses a query-data switching optimization for it). Relatedly, on \NYTimes\ at $\theta_7$, \INLJ\ is faster but at lower recall, and \MIA's recall decreases due to (1) the low connectivity of the index over highly skewed data (low ``Mode'' in Table \ref{tab:datasets}) and (2) the large per-query join size.}

\subsection{Comparison with a Dedicated OOD Index}\label{subsec:exp:revision_cdmg}

\moved{\textbf{Answer to: What if we use a SOTA OOD index for \INLJ? Why do other ideas than \MI give marginal performance gains?} Even a dedicated state-of-the-art OOD index does not close the OOD recall gap that our adaptive hybrid search closes. On \ImageNet, we strengthen \INLJ\ in two ways: replacing its BFS expansion with our hybrid search forced on for all queries (\INLJ+Hybrid), and replacing its underlying NSG index with CDMG+ \cite{CDMG}, a state-of-the-art graph index specialized for OOD search. Both improve recall over plain \INLJ\ but remain below \MIA, while \MIA\ is also faster at small thresholds (Figure \ref{fig:revision}(c)); we observe the same on \COCO\ and \LAION. This result also puts the gains of our online components in context: CDMG+ reports a $\sim$3.6$\times$ speedup over its baseline \cite{CDMG}, the same order as the $\sim$3$\times$ latency reduction of \SWS\ over \HWS\ (Section \ref{subsec:exp:base}).}

\subsection{Single-Set Search on the Merged Index}\label{subsec:exp:revision_single_set}

\moved{\textbf{Answer to: Wouldn't the merged index degrade single-set search performance?} The merged index is a \emph{join accelerator}: it complements rather than replaces per-set indexes for search-heavy workloads. Since the merged index is larger, a single-set top-$k$ search on it traverses out-set nodes before collecting enough in-set ones. We quantify this trade-off on \SIFT\ (Figure \ref{fig:revision}(d)): single-set top-$10$ search on the merged index matches the per-set index's recall, while its relative latency depends on the ratio $r=|Y|/|X\cup Y|$: $1.3\times$, $2\times$, $5\times$, and $10\times$ for $r = 0.45, 0.2, 0.06$, and $0.03$. Obviously, this relative latency decreases to $1$ as $r$ increases to $1$. This is a clean knob for practitioners. For join-heavy workloads, the modest single-set overhead is outweighed by the up to $44\times$ join speedup; for search-heavy workloads with small $r$, one keeps a per-set index alongside, raising index maintenance cost to $(1{+}r)\times$ (at most $2\times$), optionally switching between the two adaptively. Section \ref{sec:nnmap} presents our ongoing design that eliminates this trade-off altogether.}


\vspace*{-0.2cm}
\section{NN-Map: A Materialized Join Index for Vectors}\label{sec:nnmap}

\planned{\textbf{Answer to: Isn't the contribution too incremental?} We close our offline--online co-design by presenting \emph{NN-Map}. NN-Map replaces the merged index, keeps the per-set separate indexes untouched and instead materializes, per join pair, the \emph{answer} of the greedy phase itself: the top-1 nearest neighbor of each query in the other set. This retains every benefit of the merged index---threshold-agnostic reuse, offloaded in-range discovery, and the OOD signal---while eliminating its remaining trade-offs identified in Section \ref{subsec:exp:revision_single_set}: single-set searches run on unchanged per-set indexes at native speed, maintenance reduces to local updates, and supporting $N$ sets costs one integer array per joined pair rather than a union/merged graph. In short, where graph indexes navigate \emph{between vectors}, NN-Map navigates \emph{between vector sets}.}

\vspace*{-0.1cm}
\subsection{Structure and Guarantees}\label{subsec:nnmap:structure}

\vspace*{-0.05cm}
\begin{definition2}
\textbf{(NN-Map)}
\planned{Given two vector sets $X_i$ and $X_j$ from the same encoder, the NN-map $w_{i \to j}$ stores, for each $x \in X_i$, the identifier $w[x]$ of its top-1 nearest neighbor in $X_j$, the distance $dist[x] = \dist(x, w[x])$, and an OOD bit for $x$ with respect to $X_j$.}
\end{definition2}
\vspace*{-0.1cm}

\planned{Each set keeps only its ordinary \emph{home} (per-set) index $G_{X_i}$; no merged structure exists. A map is stored from the smaller set to the larger one, and the reverse join direction is served by the inverted list of the same array, so the overhead per joined pair is $\min(|X_i|, |X_j|)$ integers (plus a float and a bit). NN-Map is the vector counterpart of the classical \emph{join index} \cite{valduriez1987join}: a compact materialized access path between two relations, kept separate from either relation's own indexes. We accordingly call it a \emph{vector join index}.}

\planned{Online execution becomes a binary decision followed by the unchanged expansion phase. For a query $q$ and threshold $\theta$: if $dist[q] \geq \theta$, the result is certifiably empty and the query terminates without a single distance computation; otherwise, expansion starts directly from the in-range seed $w[q]$ in the home index $G_{X_j}$, in BFS or hybrid mode according to the OOD bit. The greedy phase does not become cheaper---it disappears.}

\vspace*{-0.05cm}
\begin{lemma2}
\textbf{(Witness Property)}
\planned{Modulo the ANN accuracy of the offline build, for every threshold $\theta$: (i) if $dist[q] \geq \theta$, then $q$ has no join match in $X_j$; (ii) otherwise, $w[q]$ is an in-range entry point for $q$. Moreover, the witness set $\{w[x]\}$ is minimal: removing the witness of any query loses the guarantee for some $(q, \theta)$.}
\end{lemma2}
\vspace*{-0.1cm}

\planned{This upgrades the high-probability entry guarantee of the merged index (Appendix C) to one that holds \emph{always} and for \emph{all} thresholds, since nearest neighbors are threshold-independent---NN-Map remains fully threshold-agnostic. The guarantee concerns entry only: matches in disconnected in-range regions remain the job of hybrid search (Section \ref{subsec:ours:ood_hybrid}), triggered by the precomputed OOD bit.}

\vspace*{-0.15cm}
\subsection{Build and Maintenance}\label{subsec:nnmap:build}

\planned{Building $w_{i \to j}$ is itself an offline batch top-1 join: each $x \in X_i$ is routed through the home index of $X_j$ with an insertion-style search that commits no edges, recording only the top-1 identifier, its distance, and the OOD score that the merged index would have derived online. The build thus costs the same asymptotic time as inserting $X_i$ into a merged index, while the memory of the merged organization disappears entirely. Notably, our own framework bootstraps this build: the batch top-1 join is executed with SWS ordering (Section \ref{subsec:ours:sws}), i.e., our join algorithm constructs its own index.}

\planned{Maintenance is local by construction, with no rebuild. Inserting $x$ into the query side is one insertion search (which returns $w[x]$ as a by-product); deleting $x$ decrements a reference count on its witness. Inserting $y$ into the data side re-checks only the reverse-nearest-neighbor queries in $y$'s vicinity; deleting $y$ re-searches only the queries whose witness was $y$, tracked by per-node owner lists. Each operation touches work proportional to the affected witnesses, directly resolving the maintenance question raised for merged structures (Section \ref{subsec:ours:mi}).}

\vspace*{-0.15cm}
\subsection{Map Selection and Composition}\label{subsec:nnmap:selection}

\planned{\textbf{Answer to: Isn't the solution limited to joins between only two sets?} With $N$ sets, which maps to materialize becomes an optimization on a set-level graph whose nodes are datasets and whose candidate edges are maps---the macro counterpart of index selection. Under a memory budget $B$, we choose a coverage $c_{ij} \in (0, 1]$ (the fraction of $X_i$ materialized) per candidate pair ($X_i, X_j$), with cost $\sum c_{ij}|X_i| \leq B$ and benefit equal to the \emph{entry-cost savings only}, since expansion is unaffected:
$\mathit{benefit}_{ij}(c) = f_{ij} \, |X_i| \left[ c \cdot \bar{g}_j + (1-c) (\bar{g}_j - \tilde{g}_j(c)) \right]$,
where $f_{ij}$ is the join frequency, $\bar{g}_j$ the average cold-start greedy cost, and $\tilde{g}_j(c)$ the residual cost of warm-starting an uncovered query from its nearest covered query via SWS (measured by a pilot run). Under partial coverage, covered queries execute first (free entry and instant empty certificates) and then serve as anchors for the rest that use our work sharing (SWS). Since earlier anchors contribute more, the benefit is concave in $c$, so a marginal-benefit greedy allocation across pairs is effective.}

\planned{Unmaterialized pairs are served by \emph{composing} stored maps: for a join $X_i \Join X_j$ routed through $X_k$, the composed seed is $s = w_{k \to j}[\,w_{i \to k}[q]\,]$, and the stored distances $d_1 = dist_{i \to k}[q]$ and $d_2 = dist_{k \to j}[w_{i \to k}[q]]$ yield, with zero additional distance computations,
$d_2 - d_1 \,\leq\, \dist(q, y^*) \,\leq\, d_1 + d_2$
for the true nearest $y^* \in X_j$; in particular, $d_2 - d_1 \geq \theta$ certifies an empty result. Longer $h$-hop compositions telescope, widening the interval by the accumulated anchor distances. Materializing a map is therefore exactly \emph{buying this uncertainty width down to zero}, which quantifies its value inside the benefit function above. Composition is also safe: seeding from the closer of the composed seed and the default entry never does worse than the baseline greedy search. Together, selection and composition make ``no $N(N{-}1)/2$ structures needed'' a property of the design rather than a claim: one home index per set, plus a budgeted subset of integer arrays, serves all pairwise joins as well as single-set searches at no extra cost.}

\vspace*{-0.15cm}
\subsection{NN-Map as Join Statistics}\label{subsec:nnmap:statistics}

\planned{The stored distances double as join statistics with an honest accuracy boundary. The count $|\{x : dist[x] < \theta\}|$ is the \emph{exact} semi-join cardinality---the number of queries with at least one match---for any $\theta$, which is precisely the question deduplication and contamination detection ask. Extending each entry to a top-$k$ prefix makes the full join size exact while $dist_k(q) > \theta$, i.e., throughout the tight-threshold regime that dominates practice (Section \ref{subsec:exp:revision_thresholds}); beyond it, a $k$NN-density extrapolation takes over: the head is exact, only the tail is modeled. This position differs from learned cardinality estimation for high-dimensional search, which featurizes distances to reference pivots \cite{lan2024cardinality}: in the known-query join setting, the references are the actual nearest neighbors and the stored distances act as answers rather than features.}

\planned{These statistics feed four consumers: (i) join \emph{direction} selection, comparing the estimated cost of either set as the outer (automating the crossover of Appendix E); (ii) join \emph{selectivity} for vector-relational optimizers, letting vector joins participate in cost-based plans alongside relational operators; (iii) per-query execution tuning, such as hybrid activation and result-buffer preallocation; and (iv) threshold recommendation, answering ``which $\theta$ yields roughly $r$ results'' by binary search on the distance CDF---the same estimate can also warn, before execution, that a chosen threshold would produce an impractically large result.
In our evaluation centering on the tight-threshold regime---where entry cost is the sole bottleneck and NN-Map removes it entirely, along with $N$-set overhead against a union graph, maintenance throughput, and the accuracy of the composition bounds.}

\vspace*{-0.2cm}
\section{Conclusion}\label{sec:conclusion}


\revisionA{}{

We presented an offline--online co-design for fast approximate vector joins in memory and proposed soft work sharing, encoder-centric merged index, and OOD-aware hybrid search to address three limitations of prior work: redundant traversal across similar queries, redundant discovery of cross-set geometry, and recall loss on difficult OOD queries. 
Our unified framework for vector join allows expressing various methods within one common pipeline, fairer comparisons, cleaner ablations, and a clearer understanding of how different design choices affect efficiency and recall. Since the framework treats the underlying graph index as a black box, it remains simple and generalizable across existing graph indexes rather than tied to a bespoke structure.
Due to its simplicity, we believe this offline--online view opens up a broader design space for future vector joins in disk-based, distributed, and vector-relational settings with filters.
Future work also includes extension to multi-table joins, non-graph-based indexes, and streaming data support.
\planned{We further propose NN-Map that materializes the join itself rather than merging the indexes, together with its map-selection optimizer and join-statistics consumers.}
}

\bibliographystyle{ACM-Reference-Format}
\bibliography{ref}

\clearpage

\section*{Appendix}
\appendix

\section{Practical Deployment Scenario}
\label{subsec:tech:deployment_scenario}

\revisionA{R2.D4}{

\moved{\textbf{Answer to: In which applications should the proposed approach be used?}} A realistic end-to-end deployment scenario for our framework is cross-source product matching in e-commerce \cite{clusterjoin, tracz2020bert, li2020deep}. Suppose a marketplace maintains a large reference catalog $Y$, while a new seller uploads a batch of product records $X$. Both catalogs are encoded by the same text or multimodal encoder \cite{SimJoin}, so their vectors already lie in a shared embedding space. The system then performs a threshold-based vector join to retrieve all candidate product pairs $(x,y)$ within distance threshold $\theta$ \cite{wang2024xling, SimJoin}, which are subsequently passed to a downstream verifier such as a rule-based matcher, a learned reranker, or human review.

This workflow highlights the practical role of our three components. First, because $X$ and $Y$ are produced by the same encoder, an encoder-centric merged index can be a natural indexing abstraction: it avoids artificially splitting a shared embedding space into separate query-side and data-side structures, and it offloads the online greedy-search effort into offline index construction. Second, seller uploads often exhibit locality, e.g., many products from the same brand, category, or style arrive together, so soft work sharing can reuse useful search effort across similar queries. Third, some products are structurally difficult, such as rare items, noisy titles, or records with sparse metadata; for these OOD-like queries, candidate matches may lie in multiple disconnected in-range regions, making OOD-aware hybrid search important for recall. Thus, even without introducing a new end-to-end benchmark, this scenario illustrates how merged indexing, soft work sharing, and hybrid traversal fit naturally into a practical shared-encoder vector-join pipeline.
\sigmod{We leave evaluation on task-specific benchmarks as future work, as our focus is on the join primitive itself.}
}

\section{Theoretical Analysis for Indexing Overhead}\label{subsec:tech:theoretical_analysis}

\revisionA{R1.D5}{
We now analyze the theoretical overhead of the merged index compared to separate indexes on $X$ and $Y$. The key question is whether jointly indexing both sets introduces a substantially larger indexing cost. We show that, under standard assumptions for graph-based indexes, the merged design preserves the same asymptotic space and time complexities. 
Therefore, the merged index should be viewed not as a heavier index structure, but as a different allocation of essentially the same indexing cost.
}

\revisionA{}{
\begin{theorem}
\textbf{(Overhead of Merged Indexing)} 
Let $X$ and $Y$ be the two vector sets to be joined, with $|X|=n_x$ and $|Y|=n_y$. Assume both the separate-index design and the merged-index design use the same neighbor-size bound $m$ per vector and other construction parameters. Then the merged index preserves the same asymptotic space and time complexities as separate indexes. 
\end{theorem}


\begin{proof}

For the space complexity:

\[
\underbrace{O(mn_x)+O(mn_y)}_{\text{separate indexes}}
\;=\;
\underbrace{O(m(n_x + n_y))}_{\text{merged index}}.
\]
Hence, the merged index adds no asymptotic space overhead over separate indexes.

The construction time of the index with $n$ vectors is
\[
T(n) = O(mn\log n) \cite{HNSW} \, \mathrm{ or } \, O\!\left(k\,n^{\frac{1+d}{d}}\log n^{1/d}+f(n)\right) \cite{NSG},
\]
which slightly differs according to the underlying index, where $d$ is the intrinsic dimensionality (we overuse the notation $d$ here), $k$ is the candidate neighborhood size used in construction, and $f(n)$ is the cost of building the initial $k$NN graph. 

Since $T$ is convex over $n$, the time complexity ratio $\alpha = \frac{T(n_x + n_y)}{T(n_x) + T(n_y)}$ between the merged and separate indexes becomes the highest when $n_x = n_y = n'$. For the former $T(n)$ above of $O(mn\log n)$,
\[
\alpha \leq \frac{T(2n')}{2T(n')} = \frac{2mn'\log 2n'}{2mn'\log n'} = \frac{\log 2n'}{\log n'} = 1 + \frac{\log 2}{\log n'}.
\]
Hence, for a sufficiently large $n'$, the second term is negligible (e.g., 0.05 for $n' = 10^6$). For the latter $T(n)$ above, when its first term subsumes the second term $f(n)$,
\[
\alpha
\leq
\frac{T(2n')}{2T(n')} =
\frac{(2n')^{1+\frac{1}{d}} \log(2n')}
{2\, n'^{1+\frac{1}{d}} \log n'}
=
2^{\frac{1}{d}} \cdot \frac{\log(2n')}{\log n'}.
\]
Assuming $d \geq 12.9$ in typical datasets \cite{NSG}, $2^{1/d}$ is lower than 1.055. Including the previously computed $\frac{\log 2n'}{\log n'}$ around 1.05, $\alpha \leq 1.11$.
\end{proof}
}

\revisionA{}{In the experiments, we show that the merged index adds negligible overheads in practice, less than 5\% even under $n_x = n_y$.}

\section{Existence of Closest Point in Local Neighborhood}
\label{subsec:tech:deployment_scenario}

\revisionA{}{

\begin{theorem}
\textbf{(Existence of top-1 closest node in local neighborhood under RNG)}
Let $u$ be a node in a relative-neighborhood graph (RNG), and let $v$ be the closest node to $u$. Then $v \in N_u$, the neighborhood of $u$.
\end{theorem}

\begin{proof}
By the RNG property \cite{RNG}, an edge $(u,v)$ exists if and only if there is no other node $w$ such that
\[
\dist(u,w) < \dist(u,v)
\quad\text{and}\quad
\dist(v,w) < \dist(v,u).
\]
Assume $v$ is the closest node to $u$. Then there cannot exist any $w$ with $\dist(u,w) < \dist(u,v)$ (Figure \ref{fig:rng}), since that would contradict the assumption that $v$ is the closest node to $u$. Hence the above condition is impossible, and therefore $(u,v)$ must be an edge in the RNG. Thus, the closest node $v$ lies in the local neighborhood $N_u$ of $u$.
\end{proof}

This theorem explains why probing the local neighborhood of a query node is often sufficient under the merged index. Existing graph indexes such as NSG are explicitly derived from RNG-style pruning, and HNSW- and Vamana-style graphs also approximate this locality-preserving behavior in practice \cite{NSG, HNSW, DiskANN}. Therefore, once a query node is inserted into the merged graph, its neighborhood is likely to expose the nearest data-side node, or at least a point very close to it. This is why the merged-index execution can often stop the greedy phase after probing the query neighborhood only.

\begin{figure}[t]
\centering
\includegraphics[width=0.4\columnwidth]{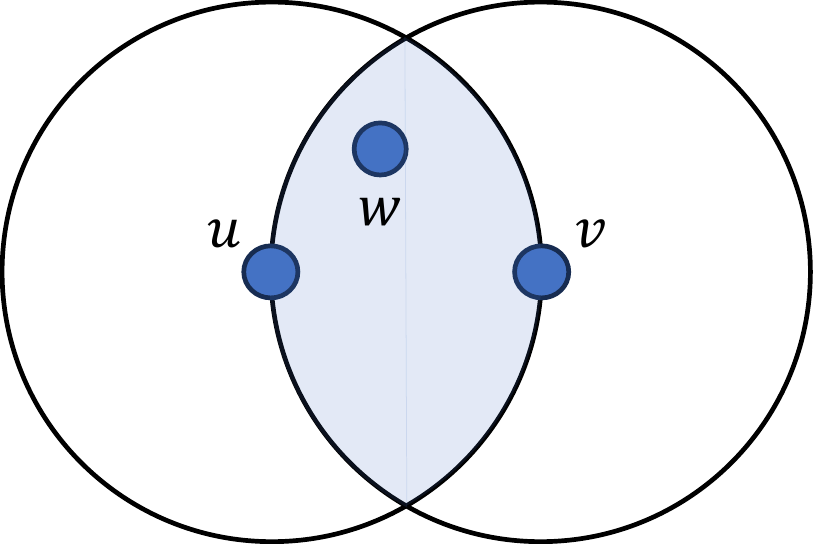}
\caption{Illustration of a common pruning rule used in graph-based indexes. $u$, $v$, and $w$ are vectors, and the circles are centered at $u$ and $v$ with radius $\dist(u,v)$. $u$ and $v$ are connected if and only if there is no $w$ in the shaded lune such that $\dist(u,w) < \dist(u,v)$ and $\dist(v,w) < \dist(v,u)$.}
\label{fig:rng}
\end{figure}

The main caveat is that an exact RNG is too expensive to construct, and practical graph indexes only approximate this property. Accordingly, a query may occasionally have no neighboring data point in its immediate local neighborhood. However, such cases are rare in practice, and modern graph indexes such as HNSW, Vamana, and NSG still preserve this local-neighborhood intuition well enough to make the merged-index design effective. This also aligns with our broader design principle of keeping the solution simple and portable across existing graph indexes, rather than introducing yet another bespoke index structure specialized to one workload.
 
}

\section{Varying Convergence Threshold}
\label{subsec:appendix:conv_threshold}

\revisionA{R1.D2}{

Figure \ref{fig:sensitivity_ES} shows that, the value of 4-16 for the \texttt{conv\_threshold} parameter in Table \ref{tab:cfg} in Section \ref{sec:ours} achieves a good latency-recall trade-off. This justifies our choice of using 10 as the default value.

\begin{figure}[t]
\centering
\includegraphics[width=0.6\columnwidth]{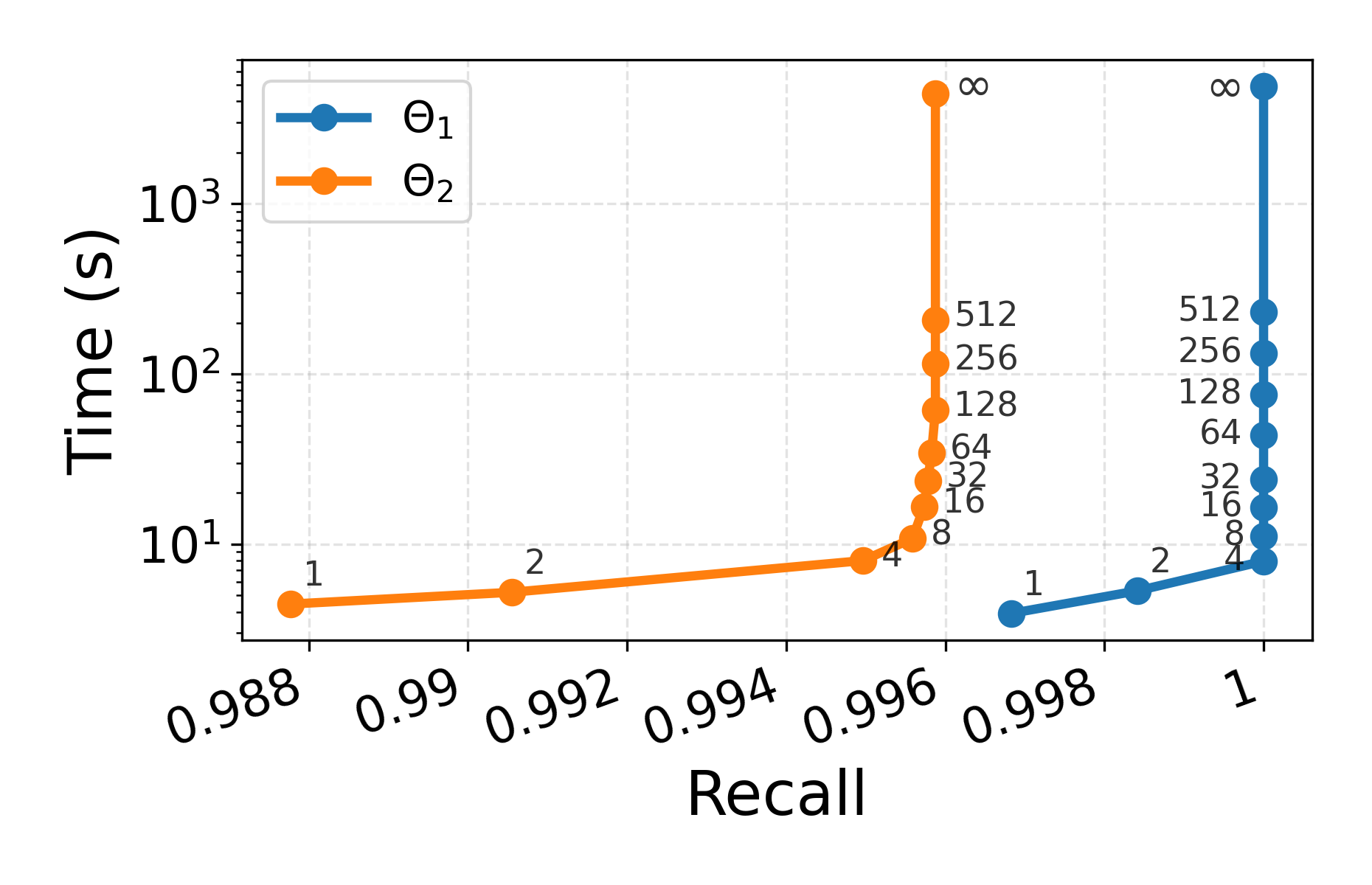}
\caption{Sensitivity analysis on varying \texttt{conv\_threshold} values (point labels) in Section \ref{sec:ours} using \SIFT dataset and two smallest thresholds.}
\label{fig:sensitivity_ES}
\end{figure}

}

\section{Discussion on Large Thresholds and Query-Data Symmetry}

\sigmod{As mentioned in Section \ref{sec:exp}, this paper focuses on small distance thresholds to align with typical top-$k$ search with small $k$ values. For large thresholds, all evaluated methods are bottlenecked by the post-greedy expansion (Figure \ref{fig:profiling_exp}). However, from the query-data symmetry for vector join, i.e., $X \Join_{dist < \theta} Y = Y \Join_{dist < \theta} X$, we observed that switching $X$ and $Y$ can reduce latency by orders of magnitude for large thresholds (and rather increase for small thresholds). This attributes to reduced random memory accesses during the post-greedy graph traversal, by sequentially iterating over the larger set $Y$ and traversing inside the smaller graph $G_X$ (or the $X$-part in $G_{X \cup Y}$). This performance gain is analogous to favoring the sequential scan over the index scan under large selectivities.}
\clearpage






\end{document}